\documentclass[preprint]{elsarticle}

\usepackage{amsmath,amssymb}
\usepackage{bm}
\usepackage{lineno,hyperref}
\usepackage{lscape}
\usepackage{algorithm}
\usepackage{algpseudocode}
\algrenewcommand\algorithmicrequire{\textbf{Input:}}
\algrenewcommand\algorithmicensure{\textbf{Output:}}
\algrenewcommand{\algorithmiccomment}[1]{\hspace{\fill}\{#1\}}
\usepackage{setspace}
\usepackage{color}
\usepackage{soul}

\newtheorem{lemma}{Lemma}
\newproof{proof}{Proof}

\biboptions{authoryear}

\begin{document}

\begin{frontmatter}

\title{Exploiting variable associations to configure efficient local search algorithms in large-scale binary integer programs\tnoteref{preliminary}}
\tnotetext[preliminary]{A preliminary version of this paper was presented in \citep{UmetaniS2015}.}

\author{Shunji Umetani\corref{cor1}}
\address{Graduate School of Information Science and Technology, Osaka University,\\ 1-5 Yamadaoka, Suita, Osaka, 565-0871, Japan}
\ead{umetani@ist.osaka-u.ac.jp}
\cortext[cor1]{Corresponding author. Tel.: +81 (6) 6879 7799.}

\begin{abstract}
We present a data mining approach for reducing the search space of local search algorithms in a class of binary integer programs including the set covering and partitioning problems.
The quality of locally optimal solutions typically improves if a larger neighborhood is used, while the computation time of searching the neighborhood increases exponentially.
To overcome this, we extract variable associations from the instance to be solved in order to identify promising pairs of flipping variables in the neighborhood search.
Based on this, we develop a 4-flip neighborhood local search algorithm that incorporates an efficient incremental evaluation of solutions and an adaptive control of penalty weights.
Computational results show that the proposed method improves the performance of the local search algorithm for large-scale set covering and partitioning problems.
\end{abstract}

\end{frontmatter}


\section{Introduction\label{sec:intro}}
The Set Covering Problem (SCP) and Set Partitioning Problem (SPP) are representative combinatorial optimization problems that have many real-world applications, such as crew scheduling \citep{BarnhartC1998,HoffmanKL1993,MingozziA1999}, vehicle routing \citep{AgarwalY1989,BaldacciR2008,BramelJ1997,HashimotoH2009}, facility location \citep{BorosE2005,FarahaniRZ2012} and logical analysis of data \citep{BorosE2000,HammerPL2006}.
Real-world applications of SCP and SPP are comprehensively reviewed in \citep{CeriaS1997} and \citep{BalasE1976}, respectively.

Given a ground set of $m$ elements $i \in M = \{ 1, \dots, m \}$, $n$ subsets $S_j \subseteq M$ ($|S_j| \ge 1$), and their costs $c_j \in \mathbb{R}$ ($\mathbb{R}$ is the set of real values) for $j \in N = \{ 1, \dots, n \}$, we say that $X \subseteq N$ is a cover of $M$ if $\bigcup_{j \in X} S_j = M$ holds.
We also say that $X \subseteq N$ is a partition of $M$ if $\bigcup_{j \in X} S_j = M$ and $S_{j_1} \cap S_{j_2} = \emptyset$ for all $j_1, j_2 \in X$ hold.
The goals of SCP and SPP are to find a minimum cost cover and partition $X$ of $M$, respectively.
In this paper, we consider the following class of Binary Integer Programs (BIPs) including SCP and SPP:
\begin{equation}
\label{eq:bip}
\begin{array}{lll}
\textnormal{minimize} & \displaystyle\sum_{j \in N} c_j x_j & \\
\textnormal{subject to} & \displaystyle\sum_{j \in N} a_{ij} x_j \le b_i, & i \in M_L,\\
 & \displaystyle\sum_{j \in N} a_{ij} x_j \ge b_i, & i \in M_G,\\
 & \displaystyle\sum_{j \in N} a_{ij} x_j = b_i, & i \in M_E,\\
& x_j \in \{ 0, 1 \}, & j \in N,
\end{array}
\end{equation}
where $a_{ij} \in \{ 0, 1 \}$ and $b_i \in \mathbb{Z}_+$ ($\mathbb{Z}_+$ is the set of nonnegative integer values).
We note that $a_{ij} = 1$ if $i \in S_j$ holds and $a_{ij} = 0$ otherwise, and $x_j = 1$ if $j \in X$ holds and $x_j = 0$ otherwise.
That is, a column vector $\bm{a}_j = (a_{1j}, \dots, a_{mj})^\top$ of the matrix $(a_{ij})$ represents the corresponding subset $S_j = \{ i \in M \mid a_{ij} = 1 \}$, and the vector $\bm{x}$ also represents the corresponding cover (or partition) $X = \{ j \in N \mid x_j = 1 \}$.
For notational convenience, we denote $M = M_L \cup M_G \cup M_E$.
For each $i \in M$, let $N_i = \{ j \in N \mid a_{ij} = 1 \}$ be the index set of subsets $S_j$ that contain the elements $i$, and let $s_i(\bm{x}) = \sum_{j \in N} a_{ij} x_j$ be the left-hand side of the $i$th constraint.

Continuous development of mathematical programming has much improved the performance of exact and heuristic algorithms and this has been accompanied by advances in computing machinery.
Many efficient exact and heuristic algorithms for large-scale SCP and SPP instances have been developed \citep{AtamturkA1995,BarahonaF2000,BastertO2010,BorndorferR1998,BoschettiMA2008,CapraraA1999,CapraraA2000,CasertaM2007,CeriaS1998,LinderothJT2001,UmetaniS2007,WedelinD1995,YagiuraM2006}, many of which are based on a variant of the \emph{column generation method} called the \emph{pricing method} that reduces the number of variables to be considered in the search by using Linear Programming (LP) and/or Lagrangian relaxation.
However, many large-scale SCP and SPP instances still remain unsolved because there is little hope of closing the large gap between the lower and upper bounds of the optimal values.
In particular, the equality constraints of SPP often make the pricing method less effective because they often prevent solutions from containing highly evaluated variables together.
In this paper, we consider an alternative approach for extracting useful features from the instance to be solved with the aim of reducing the search space of local search algorithms for large-scale SCP and SPP instances.

In the design of local search algorithms for combinatorial optimization problems, the quality of locally optimal solutions typically improves if a larger neighborhood is used.
However, the computation time of searching the neighborhood also increases exponentially.
To overcome this, extensive research has investigated ways to efficiently implement neighborhood search, which can be broadly classified into three types: (i)~reducing the number of candidates in the neighborhood \citep{PesantG1999,ShawP2002,YagiuraM1999,YagiuraM2006}, (ii)~evaluating solutions by incremental computation \citep{MichelL2000,YagiuraM1999,VanHentenryckP2005,VoudourisC2001}, and (iii)~reducing the number of variables to be considered by using LP and/or Lagrangian relaxation \citep{CapraraA1999,CeriaS1998,UmetaniS2013,YagiuraM2006}.

To suggest an alternative, we develop a data mining approach for reducing the search space of local search algorithms.
That is, we construct a $k$-nearest neighbor graph by extracting variable associations from the instance to be solved in order to identify promising pairs of flipping variables in the neighborhood search.
We also develop a 4-flip neighborhood local search algorithm that flips four variables alternately along 4-paths or 4-cycles in the $k$-nearest neighbor graph.
We incorporate an efficient incremental evaluation of solutions and an adaptive control of penalty weights into the 4-flip neighborhood local search algorithm.

\section{2-flip neighborhood local search\label{sec:local-search}}
Local Search (LS) starts from an initial solution $\bm{x}$ and then iteratively replaces $\bm{x}$ with a better solution $\bm{x}^{\prime}$ in the neighborhood $\textnormal{NB}(\bm{x})$ until no better solution is found in $\textnormal{NB}(\bm{x})$.
For some positive integer $r$, let the $r$-flip neighborhood $\textnormal{NB}_r(\bm{x})$ be the set of solutions obtainable by flipping at most $r$ variables in $\bm{x}$.
We first develop a 2-Flip Neighborhood Local Search (2-FNLS) algorithm as a basic component of our algorithm.
In order to improve efficiency, the 2-FNLS first searches $\textnormal{NB}_1(\bm{x})$, and then searches $\textnormal{NB}_2(\bm{x}) \setminus \textnormal{NB}_1(\bm{x})$ only if $\bm{x}$ is locally optimal with respect to $\textnormal{NB}_1(\bm{x})$.

The BIP is NP-hard, and the (supposedly) simpler problem of judging the existence of a feasible solution is NP-complete.
We accordingly consider the following formulation of the BIP that allows violations of the constraints and introduce the following penalized objective function with penalty weights $w_i^+ \in \mathbb{R}_+$ ($\mathbb{R}_+$ is the set of nonnegative real values) for $i \in M_L \cup M_E$ and $w_i^- \in \mathbb{R}_+$ for $i \in M_G \cup M_E$.
\begin{equation}
\label{eq:soft-bip}
\begin{array}{lll}
\textnormal{minimize} & \multicolumn{2}{l}{z(\bm{x}) = \displaystyle\sum_{j \in N} c_j x_j + \sum_{i \in M_L \cup M_E} w_i^+ y_i^+ + \sum_{i \in M_G \cup M_E} w_i^- y_i^-}\\
\textnormal{subject to} & \displaystyle\sum_{j \in N} a_{ij} x_j - y_i^+ \le b_i, & i \in M_L,\\
 & \displaystyle\sum_{j \in N} a_{ij} x_j + y_i^- \ge b_i, & i \in M_G,\\
 & \displaystyle\sum_{j \in N} a_{ij} x_j - y_i^+ + y_i^- = b_i, & i \in M_E,\\
 & x_j \in \{ 0, 1 \}, & j \in N,\\
 & y_i^+ \ge 0, & i \in M_L \cup M_E,\\
 & y_i^- \ge 0, & i \in M_G \cup M_E.\\
\end{array}
\end{equation}
For a given $\bm{x} \in \{ 0, 1 \}^n$, we can easily compute optimal $y_i^+ = | s_i(\bm{x}) - b_i |_+$ and $y_i^- = | b_i - s_i(\bm{x}) |_+$, where we denote $| x |_+ = \max \{ x , 0 \}$.

Because the region searched by a single application of LS is limited, LS is usually applied many times.
When a locally optimal solution is found, a standard strategy is to update the penalty weights and to resume LS from the obtained locally optimal solution.
We accordingly evaluate solutions by using an alternative function $\tilde{z}(\bm{x})$, where the original penalty weight vectors $\bm{w}^+$ and $\bm{w}^-$ are replaced by $\widetilde{\bm{w}}^+$ and $\widetilde{\bm{w}}^-$, respectively, and these are adaptively controlled during the search (see Section~\ref{sec:weight-control} for details).

We first describe how 2-FNLS is used to search $\textnormal{NB}_1(\bm{x})$, which is called the 1-flip neighborhood.
Let
\begin{equation}
\Delta \tilde{z}_j(\bm{x}) = \left\{
\begin{array}{ll}
\Delta \tilde{z}_j^{\uparrow}(\bm{x}) & j \in N \setminus X\\
\Delta \tilde{z}_j^{\downarrow}(\bm{x}) & j \in X,
\end{array}
\right.
\end{equation}
be the increase in $\tilde{z}(\bm{x})$ due to flipping $x_j \gets 1 - x_j$, where
\begin{equation}
\label{eq:1-flip}
\begin{array}{ll}
\Delta \tilde{z}_j^{\uparrow}(\bm{x}) = c_j + \displaystyle\sum_{i \in S_j \cap (M_L \cup M_E) \cap \{ l \mid s_l(\bm{x}) \ge b_l\}} \widetilde{w}_i^+ - \sum_{i \in S_j \cap (M_G \cup M_E) \cap \{ l \mid s_l(\bm{x}) < b_l\}} \widetilde{w}_i^-, & \\
\Delta \tilde{z}_j^{\downarrow}(\bm{x}) = - c_j - \displaystyle\sum_{i \in S_j \cap (M_L \cup M_E) \cap \{ l \mid s_l(\bm{x}) > b_l \} } \widetilde{w}_i^+ + \sum_{i \in S_j \cap (M_G \cup M_E) \cap \{ l \mid s_l(\bm{x}) \le b_l \} } \widetilde{w}_i^-,
\end{array}
\end{equation}
are the increases in $\tilde{z}(\bm{x})$ due to flipping $x_j = 0 \to 1$ and $x_j = 1 \to 0$, respectively.
2-FNLS first searches for an improved solution obtainable by flipping $x_j \gets 1 - x_j$ for $j \in N$.
If an improved solution exists, it chooses $j$ with the minimum value of $\Delta \tilde{z}_j(\bm{x})$ and flips $x_j \gets 1 - x_j$.

We next describe how 2-FNLS is used to search $\textnormal{NB}_2(\bm{x}) \setminus \textnormal{NB}_1(\bm{x})$, which is called the 2-flip neighborhood.
We derive conditions that reduce the number of candidates in $\textnormal{NB}_2(\bm{x}) \setminus \textnormal{NB}_1(\bm{x})$ without sacrificing the solution quality by expanding the results as shown in \citep{YagiuraM2006}.
Let $\Delta \tilde{z}_{j_1,j_2}(\bm{x})$ be the increase in $\tilde{z}(\bm{x})$ due to simultaneously flipping the values of $x_{j_1}$ and $x_{j_2}$.
\begin{lemma}
\label{lem:nb1}
If a solution $\bm{x}$ is locally optimal with respect to $\textnormal{NB}_1(\bm{x})$, then $\Delta \tilde{z}_{j_1,j_2}(\bm{x}) < 0$ holds only if $S_{j_1} \cap S_{j_2} \not= \emptyset$ and $x_{j_1} \not= x_{j_2}$.
\end{lemma}
\begin{proof}
By the assumption of the lemma, $\Delta \tilde{z}_{j_1}(\bm{x}) \ge 0$ and $\Delta \tilde{z}_{j_2}(\bm{x}) \ge 0$ hold.
It is clear from (\ref{eq:1-flip}) that $\Delta \tilde{z}_{j_1,j_2} = \Delta \tilde{z}_{j_1}(\bm{x}) + \Delta \tilde{z}_{j_2}(\bm{x}) \ge 0$ holds if $S_{j_1} \cap S_{j_2} = \emptyset$.

We show that $\Delta \tilde{z}_{j_1,j_2}(\bm{x}) \ge 0$ holds if $S_{j_1} \cap S_{j_2} \not= \emptyset$ and $x_{j_1} = x_{j_2}$.
First, we consider the case of $x_{j_1} = x_{j_2} = 1$.
If $s_i(\bm{x}) = b_i + 1$ holds for $i \in S_{j_1} \cap S_{j_2} \cap (M_L \cup M_E)$, then decrease of the violation $y_i^+$ partly cancels by flipping $x_{j_1} = 1 \to 0$ and $x_{j_2} = 1 \to 0$ simultaneously.
Similarly, if $s_i(\bm{x}) = b_i + 1$ holds for $i \in S_{j_1} \cap S_{j_2} \cap (M_G \cup M_E)$, then a new violation $y_i^-$ occurs by flipping $x_{j_1} = 1 \to 0$ and $x_{j_2} = 1 \to 0$ simultaneously.
We then have
\begin{multline}
\Delta \tilde{z}_{j_1,j_2}(\bm{x}) = \Delta \tilde{z}_{j_1}^{\downarrow}(\bm{x}) + \Delta \tilde{z}_{j_2}^{\downarrow}(\bm{x}) + \displaystyle\sum_{i \in S_{j_1} \cap S_{j_2} \cap (M_L \cup M_E) \cap \{ l \mid s_l(\bm{x}) = b_l + 1\}} \widetilde{w}_i^+ \\
 + \displaystyle\sum_{i \in S_{j_1} \cap S_{j_2} \cap (M_G \cup M_E) \cap \{ l \mid s_l(\bm{x}) = b_l + 1\}} \widetilde{w}_i^- \ge 0.
\end{multline}

Next, we consider the case of $x_{j_1} = x_{j_2} = 0$.
If $s_i(\bm{x}) = b_i - 1$ holds for $i \in S_{j_1} \cap S_{j_2} \cap (M_L \cup M_E)$, then a new violation $y_i^+$ occurs by flipping $x_{j_1} = 0 \to 1$ and $x_{j_2} = 0 \to 1$ simultaneously.
Similarly, if $s_i(\bm{x}) = b_i - 1$ holds for $i \in S_{j_1} \cap S_{j_2} \cap (M_G \cup M_E)$, then decrease of the violation $y_i^-$ partly cancels by flipping $x_{j_1} = 0 \to 1$ and $x_{j_2} = 0 \to 1$ simultaneously.
We then have
\begin{multline}
\Delta \tilde{z}_{j_1,j_2}(\bm{x}) = \Delta \tilde{z}_{j_1}^{\uparrow}(\bm{x}) + \Delta \tilde{z}_{j_2}^{\uparrow}(\bm{x}) + \displaystyle\sum_{i \in S_{j_1} \cap S_{j_2} \cap (M_L \cup M_E) \cap \{ l \mid s_l(\bm{x}) = b_l - 1 \}} \widetilde{w}_i^+\\
 + \displaystyle\sum_{i \in S_{j_1} \cap S_{j_2} \cap (M_G \cup M_E) \cap \{ l \mid s_l(\bm{x}) = b_l - 1 \}} \widetilde{w}_i^- \ge 0.
\end{multline}
\qed
\end{proof}

Based on this lemma, we consider only the case of $x_{j_1} = 1$ and $x_{j_2} = 0$.
If $s_i(\bm{x}) = b_i$ holds for $i \in S_{j_1} \cap S_{j_2} \cap (M_L \cup M_E)$, then increase of the violation $y_i^+$ by flipping $x_{j_2} = 0 \to 1$ cancels by flipping $x_{j_1} = 1 \to 0$ simultaneously, while no decrease of the violation $y_i^+$ occurs by flipping $x_{j_1} = 1 \to 0$ independently.
Similarly, if $s_i(\bm{x}) = b_i$ holds for $i \in S_{j_1} \cap S_{j_2} \cap (M_G \cup M_E)$, then increase of the violation $y_i^-$ by flipping $x_{j_1} = 1 \to 0$ cancels by flipping $x_{j_2} = 0 \to 1$ simultaneously, while no decrease of the violation $y_i^-$ occurs by flipping $x_{j_2} = 0 \to 1$ independently.
We then have
\begin{equation}
\label{eq:2-flip}
\Delta \tilde{z}_{j_1,j_2}(\bm{x}) = \Delta \tilde{z}^{\downarrow}_{j_1}(\bm{x}) + \Delta \tilde{z}^{\uparrow}_{j_2}(\bm{x}) - \displaystyle\sum_{i \in \bar{S}(\bm{x}) \cap (M_L \cup M_E)} \widetilde{w}^+_i - \sum_{i \in \bar{S}(\bm{x}) \cap (M_G \cup M_E)} \widetilde{w}^-_i,
\end{equation}
where $\bar{S}(\bm{x}) = \{ i \in S_{j_1} \cap S_{j_2} \mid s_i(\bm{x}) = b_i \}$.
From these results, the 2-flip neighborhood can be restricted to the set of solutions satisfying $x_{j_1} \not= x_{j_2}$ and $\bar{S}(\bm{x}) \not= \emptyset$.
However, it might not be possible to search this set efficiently without first extracting it.
We thus construct a neighbor list that stores promising pairs of variables $x_{j_1}$ and $x_{j_2}$ for efficiency (see Section \ref{sec:variable-association} for details).

To increase the efficiency of 2-FNLS, we decompose the neighborhood $\textnormal{NB}_2(\bm{x})$ into a number of sub-neighborhoods.
Let $\textnormal{NB}_2^{(j_1)}(\bm{x}) = \{ \bm{x}^{\prime} \in \textnormal{NB}_2(\bm{x}) \mid x_{j_1} = 1, x_{j_1}^{\prime} = 0 \}$ be the subset of $\textnormal{NB}_2(\bm{x})$ obtainable by flipping $x_{j_1} = 1 \to 0$.
2-FNLS searches $\textnormal{NB}_2^{(j_1)}(\bm{x})$ for each $j_1 \in X$ in ascending order of $\Delta \tilde{z}_{j_1}^{\downarrow}(\bm{x})$.
If an improved solution exists, it chooses the pair $j_1$ and $j_2$ with the minimum value of $\Delta \tilde{z}_{j_1,j_2}(\bm{x})$ among $\textnormal{NB}_2^{(j_1)}(\bm{x})$ and flips $x_{j_1} = 1 \to 0$ and $x_{j_2} = 0 \to 1$.
2-FNLS skips to search the remaining sub-neighborhoods and immediately return to search $\textnormal{NB}_1(\bm{x})$ whenever an improved solution is obtained in a sub-neighborhood $\textnormal{NB}_2^{(j_1)}(\bm{x})$.
The first version of 2-FNLS is formally described as Algorithm~\ref{alg:2fnls}.

\begin{algorithm}
\caption{2-FNLS($\bm{x},\widetilde{\bm{w}}^+,\widetilde{\bm{w}}^-$)\label{alg:2fnls}}
{\small
\begin{spacing}{1.0}
\begin{algorithmic}[1]
\Require A solution $\bm{x}$ and penalty weight vectors $\widetilde{\bm{w}}^+$ and $\widetilde{\bm{w}}^-$.
\Ensure A solution $\bm{x}$.
\Statex
\State \textbf{START:}
\State $\bm{x}^{\prime} \gets \bm{x}$
\For{$j \in N$} \Comment{Search $\textnormal{NB}_1(\bm{x})$}
\If{$\tilde{z}(\bm{x}) + \Delta \tilde{z}_j(\bm{x}) < \tilde{z}(\bm{x}^{\prime})$}
\State $\bm{x}^{\prime} \gets \bm{x}$, $x_j^{\prime} \gets 1 - x_j^{\prime}$
\EndIf
\EndFor
\If{$\tilde{z}(\bm{x}^{\prime}) < \tilde{z}(\bm{x})$}
\State $\bm{x} \gets \bm{x}^{\prime}$
\State \textbf{goto START}
\EndIf
\For{$j_1 \in X$ in ascending order of $\Delta \tilde{z}_{j_1}^{\downarrow}(\bm{x})$} \Comment{Search $\textnormal{NB}_2(\bm{x})$}
\State $\bm{x}^{\prime} \gets \bm{x}$
\For{$j_2 \in N \setminus X$} \Comment{Search $\textnormal{NB}_2^{(j_1)}(\bm{x})$}
\If{$\tilde{z}(\bm{x}) + \Delta \tilde{z}_{j_1,j_2}(\bm{x}) < \tilde{z}(\bm{x}^{\prime})$}
\State $\bm{x}^{\prime} \gets \bm{x}$, $x_{j_1}^{\prime} \gets 0$, $x_{j_2}^{\prime} \gets 1$
\EndIf
\EndFor
\If{$\tilde{z}(\bm{x}^{\prime}) < \tilde{z}(\bm{x})$}
\State $\bm{x} \gets \bm{x}^{\prime}$
\State \textbf{goto START}
\EndIf
\EndFor
\end{algorithmic}
\end{spacing}
}
\end{algorithm}

\section{Efficient incremental evaluation\label{sec:eval}}
The 2-FNLS requires $\textnormal{O}(\sigma)$ time to compute the value of the evaluation function $\tilde{z}(\bm{x})$ for the current solution $\bm{x}$ if implemented naively, where $\sigma = \sum_{i \in M} \sum_{j \in N} a_{ij}$ denote the number of nonzero elements in the constraint matrix $(a_{ij})$.
To overcome this, we first develop a standard incremental evaluation of $\Delta \tilde{z}_j^{\uparrow}(\bm{x})$ and $\Delta \tilde{z}_j^{\downarrow}(\bm{x})$ in $\textnormal{O}(|S_j|)$ time by keeping the values of the left-hand side of constraints $s_i(\bm{x})$ for $i \in M$ in memory.
We further develop an improved incremental evaluation of $\Delta \tilde{z}_j^{\uparrow}(\bm{x})$ and $\Delta \tilde{z}_j^{\downarrow}(\bm{x})$ in $\textnormal{O}(1)$ time by keeping additional auxiliary data in memory.
By using this, 2-FNLS is also able to evaluate $\Delta\tilde{z}_{j_1,j_2}(\bm{x})$ in $\textnormal{O}(|S_j|)$ time by (\ref{eq:2-flip}).

We first consider a standard incremental evaluation of $\Delta \tilde{z}_j^{\uparrow}(\bm{x})$ and $\Delta \tilde{z}_j^{\downarrow}(\bm{x})$ in $\textnormal{O}(|S_j|)$ time using the following formulas:
\begin{equation}
\begin{array}{lcl}
\Delta \tilde{z}_j^{\uparrow}(\bm{x}) & = & c_j + \Delta \tilde{p}_j^{\uparrow}(\bm{x}) + \Delta \tilde{q}_j^{\uparrow}(\bm{x}),\\
\Delta \tilde{p}_j^{\uparrow}(\bm{x}) & = & \displaystyle\sum_{i \in S_j \cap (M_L \cup M_E)} \widetilde{w}_i^+ \left( |(s_i(\bm{x})+1)-b_i|_+ - |s_i(\bm{x}) - b_i|_+ \right),\\
\Delta \tilde{q}_j^{\uparrow}(\bm{x}) & = & \displaystyle\sum_{i \in S_j \cap (M_G \cup M_E)} \widetilde{w}_i^- \left( |b_i - (s_i(\bm{x})+1)|_+ - |b_i - s_i(\bm{x})|_+ \right),
\end{array}
\end{equation}
\begin{equation}
\begin{array}{lcl}
\Delta \tilde{z}_j^{\downarrow}(\bm{x}) & = & - c_j + \Delta \tilde{p}_j^{\downarrow}(\bm{x}) + \Delta \tilde{q}_j^{\downarrow}(\bm{x}),\\
\Delta \tilde{p}_j^{\downarrow}(\bm{x}) & = & \displaystyle\sum_{i \in S_j \cap (M_L \cup M_E)} \widetilde{w}_i^+ \left( |(s_i(\bm{x})-1)-b_i|_+ - |s_i(\bm{x}) - b_i|_+ \right),\\
\Delta \tilde{q}_j^{\downarrow}(\bm{x}) & = & \displaystyle\sum_{i \in S_j \cap (M_G \cup M_E)} \widetilde{w}_i^- \left( |b_i - (s_i(\bm{x})-1)|_+ - |b_i - s_i(\bm{x})|_+ \right),
\end{array}
\end{equation}
where 2-FNLS keeps the values of the left-hand side of constraints $s_i(\bm{x})$ for $i \in M$ in memory.
2-FNLS updates $s_i(\bm{x})$ for $i \in S_j$ in $\textnormal{O}(|S_j|)$ time by $s_i(\bm{x}^{\prime}) \gets s_i(\bm{x}) + 1$ and $s_i(\bm{x}^{\prime}) \gets s_i(\bm{x}) - 1$ when the current solution $\bm{x}$ moves to $\bm{x}^{\prime}$ by flipping $x_j = 0 \to 1$ and $x_j = 1 \to 0$, respectively.

We further develop an improved incremental evaluation of $\Delta \tilde{z}_j^{\uparrow}(\bm{x})$ and $\Delta \tilde{z}_j^{\downarrow}(\bm{x})$ in $\textnormal{O}(1)$ time by directly keeping $\Delta \tilde{p}_j^{\uparrow}(\bm{x})$, $\Delta \tilde{q}_j^{\uparrow}(\bm{x})$ for $j \in N \setminus X$ and $\Delta \tilde{p}_j^{\downarrow}(\bm{x})$, $\Delta \tilde{q}_j^{\downarrow}(\bm{x})$ for $j \in X$ in memory.
When the current solution $\bm{x}$ moves to $\bm{x}^{\prime}$ by flipping $x_j = 0 \to 1$, 2-FNLS first updates $s_i(\bm{x})$ for $i \in S_j$ in $\textnormal{O}(|S_j|)$ time by $s_i(\bm{x}^{\prime}) \gets s_i(\bm{x}) + 1$, and then updates $\Delta \tilde{p}_k^{\uparrow}(\bm{x})$, $\Delta \tilde{q}_k^{\uparrow}(\bm{x})$ for $k \in N_i \setminus X$, $i \in S_j$ and $\Delta \tilde{p}_k^{\downarrow}(\bm{x})$, $\Delta \tilde{q}_k^{\downarrow}(\bm{x})$ for $k \in N_i \cap X$, $i \in S_j$ in $\textnormal{O}(\sum_{i \in S_j}|N_i|)$ time using the following formulas:
\begin{equation}
\begin{array}{lcl}
\Delta \tilde{p}_k^{\uparrow}(\bm{x}^{\prime}) & \gets & \Delta \tilde{p}_k^{\uparrow}(\bm{x}) + \displaystyle\sum_{i \in S_j \cap S_k \cap (M_L \cup M_E)} \widetilde{w}_i^+ \left( \Delta y_i^+(\bm{x}^{\prime}) - \Delta y_i^+(\bm{x}) \right),\\
\Delta \tilde{q}_k^{\uparrow}(\bm{x}^{\prime}) & \gets & \Delta \tilde{q}_k^{\uparrow}(\bm{x}) + \displaystyle\sum_{i \in S_j \cap S_k \cap (M_G \cup M_E)} \widetilde{w}_i^- \left( \Delta y_i^-(\bm{x}^{\prime}) - \Delta y_i^-(\bm{x}) \right),\\
\Delta \tilde{p}_k^{\downarrow}(\bm{x}^{\prime}) & \gets & \Delta \tilde{p}_k^{\downarrow}(\bm{x}) + \displaystyle\sum_{i \in S_j \cap S_k \cap (M_L \cup M_E)} \widetilde{w}_i^+ \left( \Delta y_i^+(\bm{x}^{\prime}) - \Delta y_i^+(\bm{x}) \right),\\
\Delta \tilde{q}_k^{\downarrow}(\bm{x}^{\prime}) & \gets & \Delta \tilde{q}_k^{\downarrow}(\bm{x}) + \displaystyle\sum_{i \in S_j \cap S_k \cap (M_G \cup M_E)} \widetilde{w}_i^- \left( \Delta y_i^-(\bm{x}^{\prime}) - \Delta y_i^-(\bm{x}) \right),
\end{array}
\end{equation}
where
\begin{equation}
\begin{array}{lcl}
\Delta y_i^+(\bm{x}^{\prime}) & = & | (s_i(\bm{x}^{\prime}) + 1) - b_i |_+ - |s_i(\bm{x}^{\prime}) - b_i|_+,\\
\Delta y_i^+(\bm{x}) & = & |s_i(\bm{x}^{\prime}) - b_i|_+ - |s_i(\bm{x}) - b_i|_+,\\
\Delta y_i^-(\bm{x}^{\prime}) & = & | b_i - (s_i(\bm{x}^{\prime}) + 1) |_+ - |b_i - s_i(\bm{x}^{\prime})|_+,\\
\Delta y_i^-(\bm{x}) & = & |b_i - s_i(\bm{x}^{\prime})|_+ - |b_i - s_i(\bm{x})|_+.
\end{array}
\end{equation}
Similarly, when the current solution $\bm{x}$ moves to $\bm{x}^{\prime}$ by flipping $x_j = 1 \to 0$, 2-FNLS first updates $s_i(\bm{x})$ for $i \in S_j$ in $\textnormal{O}(|S_j|)$ time, and then updates $\Delta \tilde{p}_k^{\uparrow}(\bm{x})$, $\Delta \tilde{q}_k^{\uparrow}(\bm{x})$ for $k \in N_i \setminus X$, $i \in S_j$ and $\Delta \tilde{p}_k^{\downarrow}(\bm{x})$, $\Delta \tilde{q}_k^{\downarrow}(\bm{x})$ for $k \in N_i \cap X$, $i \in S_j$ in $\textnormal{O}(\sum_{i \in S_j}|N_i|)$ time.
We note that the computation time for updating the auxiliary data has little effect on the total computation time of 2-FNLS because the number of solutions actually visited is much less than the number of neighbor solutions evaluated in most cases.

\section{Exploiting variable associations\label{sec:variable-association}}
It is known that the quality of locally optimal solutions improves if a larger neighborhood is used.
However, the computation time to search the neighborhood $\textnormal{NB}_r(\bm{x})$ also increases exponentially with $r$, since $|\textnormal{NB}_r(\bm{x})| = \textnormal{O}(n^r)$ generally holds.
A large amount of computation time is thus needed in practice in order to scan all candidates in $\textnormal{NB}_2(\bm{x})$ for large-scale instances with millions of variables.
To overcome this, we develop a data mining approach that identifies promising pairs of flipping variables in $\textnormal{NB}_2(\bm{x})$ by extracting variable associations from the instance to be solved using only a small amount of computation time.

\begin{figure}[tb]
\centering
\begin{minipage}{0.48\textwidth}
{\scriptsize
\[
(a_{ij}) = \left(
\begin{array}{cccccccccc}
0 & 1 & 1 & 1 & 0 & 0 & 0 & 0 & 1 & 0\\
1 & 0 & 1 & 0 & 1 & 0 & 1 & 0 & 0 & 1\\
1 & 1 & 0 & 0 & 0 & 0 & 1 & 1 & 0 & 1\\
0 & 0 & 1 & 0 & 1 & 1 & 1 & 1 & 1 & 0\\
1 & 1 & 0 & 1 & 0 & 1 & 0 & 1 & 0 & 0\\
0 & 0 & 0 & 1 & 1 & 0 & 0 & 0 & 1 & 1
\end{array}
\right)
\longrightarrow
\]
}
\end{minipage}
\hfill
\begin{minipage}{0.48\textwidth}
\centering
\includegraphics[width=0.95\textwidth]{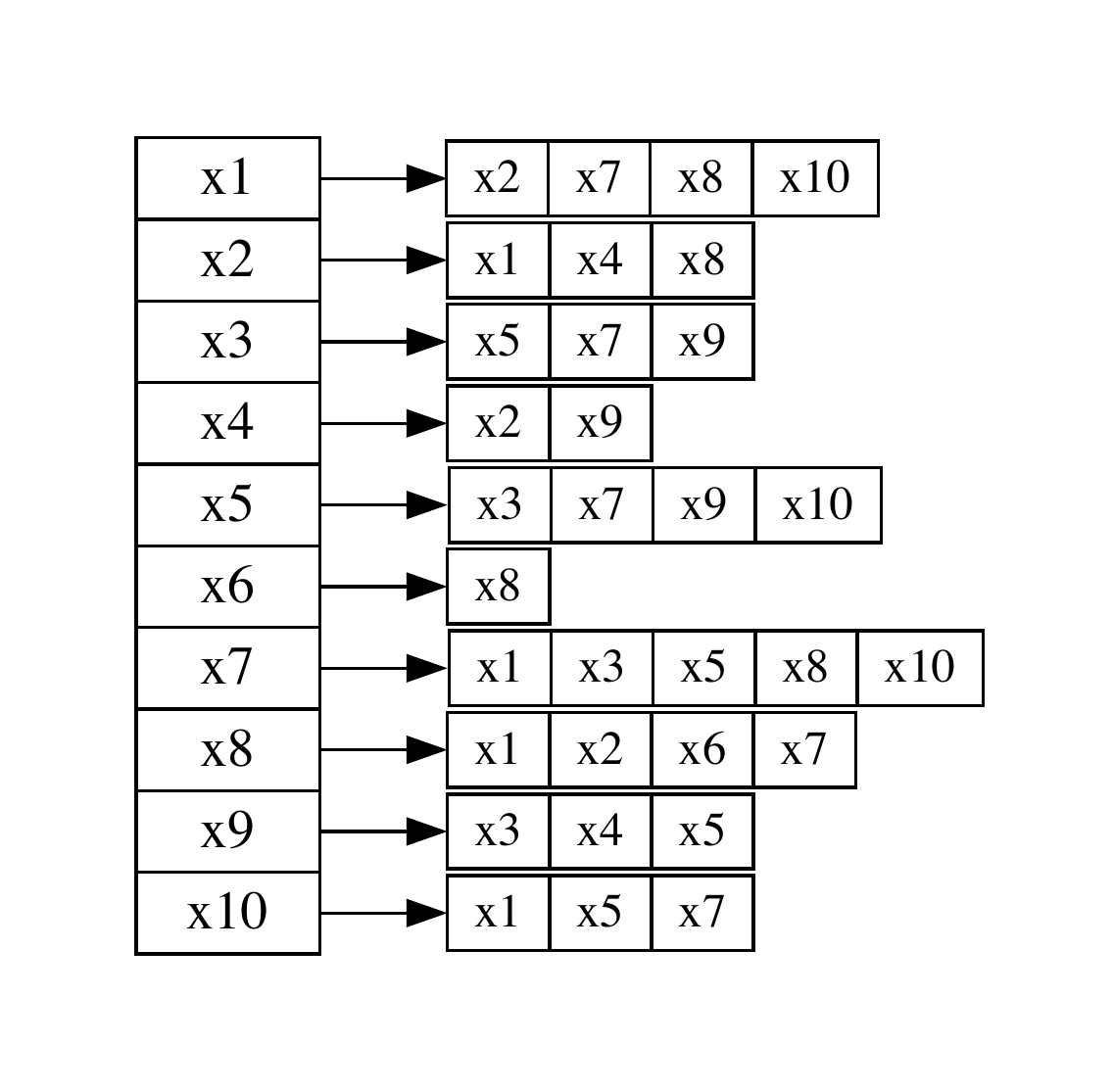}
\end{minipage}
\caption{Example of the neighbor list\label{fig:nb-list}}
\end{figure}

From the results in Section~\ref{sec:local-search}, the 2-flip neighborhood can be restricted to the set of solutions satisfying $x_{j_1} \not= x_{j_2}$ and $\bar{S}(\bm{x}) \not= \emptyset$.
We further observe from (\ref{eq:2-flip}) that it is favorable to select pairs of flipping variables $x_{j_1}$ and $x_{j_2}$ which gives a larger size $|S_{j_1} \cap S_{j_2}|$ in order to obtain $\Delta \tilde{z}_{j_1,j_2}(\bm{x}) < 0$.
Based on this observation, we keep a limited set of pairs of variables $x_{j_1}$ and $x_{j_2}$ for which $|S_{j_1} \cap S_{j_2}|$ is large in memory, and we call this the \emph{neighbor list} (Figure~\ref{fig:nb-list}).
We note that $|S_{j_1} \cap S_{j_2}|$ represents a certain kind of similarity between the subsets $S_{j_1}$ and $S_{j_2}$ (or column vectors $\bm{a}_{j_1}$ and $\bm{a}_{j_2}$ of the constraint matrix $(a_{ij})$) and we keep the $k$-nearest neighbors for each variable $x_j$ for $j \in N$ in the neighbor list.
(In Figure~\ref{fig:nb-list}, we keep the set of pairs of variables $x_{j_1}$ and $x_{j_2}$ having $|S_{j_1} \cap S_{j_2}| \ge 2$ in the neighbor list, because small examples often have many ties.)

For each variable $x_{j_1}$ for $j_1 \in N$, we store the first $k = \min \{ |N^{(j_1)}|, \alpha |M| \}$ variables $x_{j_2}$ ($j_2 \not= j_1$) in descending order of $|S_{j_1} \cap S_{j_2}|$ in the $j_1$th row of the neighbor list, where $N^{(j_1)} = \{ j_2 \in N \mid j_2 \not= j_1, S_{j_1} \cap S_{j_2} \not= \emptyset \}$ and $\alpha$ is a program parameter that we set to five.
Let $L[j_1]$ be the index set of variables $x_{j_2}$ stored in the $j_1$th row of the neighbor list.
We then reduce the number of candidates in $\textnormal{NB}_2(\bm{x})$ by restricting the pairs of flipping variables $x_{j_1}$ and $x_{j_2}$ to pairs in the neighbor list $j_1 \in X$ and $j_2 \in (N \setminus X) \cap L[j_1]$.

We note that it is still expensive to construct the whole neighbor list for large-scale instances with millions of variables.
To overcome this, we develop a lazy construction algorithm for the neighbor list.
That is, 2-FNLS starts from an empty neighbor list and generates the $j_1$th row of the neighbor list $L[j_1]$ only when 2-FNLS searches $\textnormal{NB}_2^{(j_1)}(\bm{x})$ for the first time.
The improved version of 2-FNLS is formally described as Algorithm~\ref{alg:mod-2fnls}.

\begin{algorithm}
\caption{2-FNLS($\bm{x},\widetilde{\bm{w}}^+,\widetilde{\bm{w}}^-$)\label{alg:mod-2fnls}}
{\small
\begin{spacing}{1.0}
\begin{algorithmic}[1]
\Require A solution $\bm{x}$ and penalty weight vectors $\widetilde{\bm{w}}^+$ and $\widetilde{\bm{w}}^-$.
\Ensure A solution $\bm{x}$.
\Statex
\For{$j_1 \in N$} \Comment{Initialize $L[j_1]$}
\State Set $L[j_1] \gets \emptyset$
\EndFor
\State \textbf{START:}
\State $\bm{x}^{\prime} \gets \bm{x}$
\For{$j \in N$} \Comment{Search $\textnormal{NB}_1(\bm{x})$}
\If{$\tilde{z}(\bm{x}) + \Delta \tilde{z}_j(\bm{x}) < \tilde{z}(\bm{x}^{\prime})$}
\State $\bm{x}^{\prime} \gets \bm{x}$, $x_j^{\prime} \gets 1 - x_j^{\prime}$
\EndIf
\EndFor
\If{$\tilde{z}(\bm{x}^{\prime}) < \tilde{z}(\bm{x})$}
\State $\bm{x} \gets \bm{x}^{\prime}$
\State \textbf{goto START}
\EndIf
\For{$j_1 \in X$ in ascending order of $\Delta \tilde{z}_{j_1}^{\downarrow}(\bm{x})$} \Comment{Search $\textnormal{NB}_2(\bm{x})$}
\If{$L[j_1] = \emptyset$} \Comment{Generate $L[j_1]$}
\State Generate $L[j_1]$
\EndIf
\State $\bm{x}^{\prime} \gets \bm{x}$
\For{$j_2 \in (N \setminus X) \cap L[j_1]$} \Comment{Search $\textnormal{NB}_2^{(j_1)}(\bm{x})$}
\If{$\tilde{z}(\bm{x}) + \Delta \tilde{z}_{j_1,j_2}(\bm{x}) < \tilde{z}(\bm{x}^{\prime})$}
\State $\bm{x}^{\prime} \gets \bm{x}$, $x_{j_1}^{\prime} \gets 0$, $x_{j_2}^{\prime} \gets 1$
\EndIf
\EndFor
\If{$\tilde{z}(\bm{x}^{\prime}) < \tilde{z}(\bm{x})$}
\State $\bm{x} \gets \bm{x}^{\prime}$
\State \textbf{goto START}
\EndIf
\EndFor
\end{algorithmic}
\end{spacing}
}
\end{algorithm}

A similar approach has been developed in local search algorithms for the Euclidean Traveling Salesman Problem (TSP) in which a sorted list containing only the $k$-nearest neighbors is stored for each city by using a geometric data structure called the $k$-dimensional tree \citep{JohnsonDS1997}.
However, it is not suitable for finding the $k$-nearest neighbors efficiently in high-dimensional spaces.
We thus extend it for application to the high-dimensional column vectors $\bm{a}_j \in \{ 0, 1 \}^m$ for $j \in N$ of BIPs by using a lazy construction algorithm for the neighbor list.

\section{4-flip neighborhood local search\label{sec:4-flip}}
We can regard the neighbor-list in Section~\ref{sec:variable-association} as an adjacency-list representation of a directed graph, and represent associations between variables by a corresponding directed graph called the \emph{$k$-nearest neighbor graph}.
Figure~\ref{fig:graph} illustrates an example of the $k$-nearest neighbor graph corresponding to the neighbor-list in Figure~\ref{fig:nb-list}.

\begin{figure}[tb]
\centering
\includegraphics[width=0.48\textwidth]{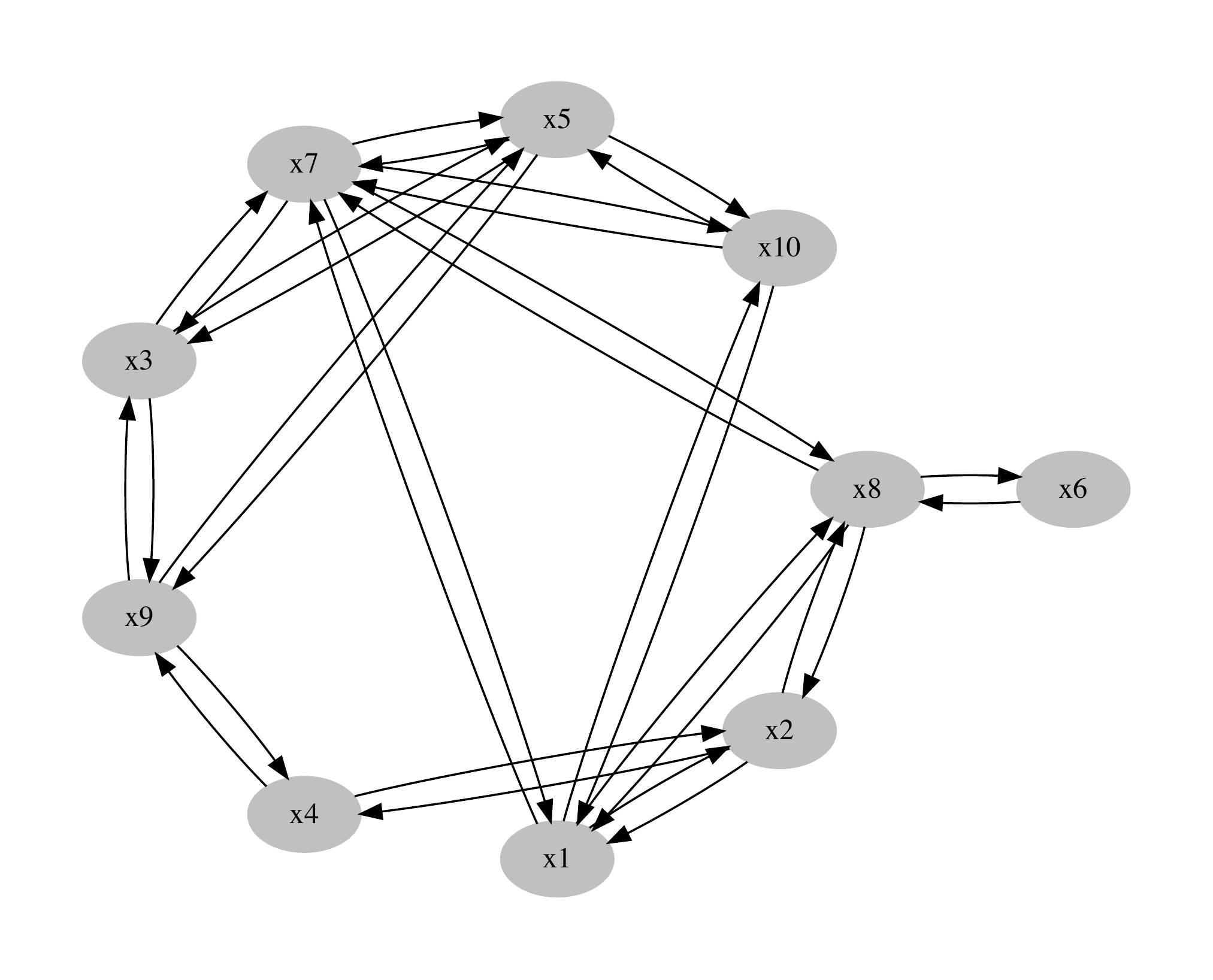}
\caption{Example of the $k$-nearest neighbor graph\label{fig:graph}}
\end{figure}

Using the $k$-nearest neighbor graph, we extend 2-FNLS to search a set of promising neighbor solutions in $\textnormal{NB}_4(\bm{x})$.
For each variable $x_{j_1}$ for $j_1 \in X$, we keep the variables $x_{j_2}$ for $j_2 \in (N \setminus X) \cap L[j_1]$ having the minimum value of $\Delta \tilde{z}_{j_1,j_2}(\bm{x})$ in memory as $j_2 = \pi(j_1)$.
The extended 2-FNLS, called the 4-Flip Neighborhood Local Search (4-FNLS) algorithm, then searches for an improved solution by flipping $x_{j_1} = 1 \to 0$, $x_{\pi(j_1)} = 0 \to 1$, $x_{j_3} = 1 \to 0$ and $x_{\pi(j_3)} = 0 \to 1$ for $j_1 \in X$ and $j_3 \in X \cap L[\pi(j_1)]$ satisfying $j_1 \not= j_3$ and $\pi(j_1) \not= \pi(j_3)$, i.e., flipping the values of variables alternately along 4-paths or 4-cycles in the $k$-nearest neighbor graph.
Let $\Delta \tilde{z}_{j_1,j_2,j_3,j_4}(\bm{x})$ be the increase in $\tilde{z}(\bm{x})$ due to simultaneously flipping $x_{j_1} = 1 \to 0$, $x_{j_2} = 0 \to 1$, $x_{j_3} = 1 \to 0$ and $x_{j_4} = 0 \to 1$.
4-FNLS computes $\Delta \tilde{z}_{j_1,j_2,j_3,j_4}(\bm{x})$ in $\textnormal{O}(|S_j|)$ time by applying the standard incremental evaluation alternately (see Section~\ref{sec:eval} for details).
4-FNLS is formally described with as Algorithm~\ref{alg:4fnls}.

\begin{algorithm}
\caption{4-FNLS($\bm{x},\widetilde{\bm{w}}^+,\widetilde{\bm{w}}^-$)\label{alg:4fnls}}
{\small
\begin{spacing}{1.0}
\begin{algorithmic}[1]
\Require A solution $\bm{x}$ and penalty weight vectors $\widetilde{\bm{w}}^+$ and $\widetilde{\bm{w}}^-$.
\Ensure A solution $\bm{x}$.
\Statex
\For{$j_1 \in N$} \Comment{Initialize $L[j_1]$}
\State Set $L[j_1] \gets \emptyset$
\EndFor
\State \textbf{START:}
\State $\bm{x}^{\prime} \gets \bm{x}$
\For{$j \in N$} \Comment{Search $\textnormal{NB}_1(\bm{x})$}
\If{$\tilde{z}(\bm{x}) + \Delta \tilde{z}_j(\bm{x}) < \tilde{z}(\bm{x}^{\prime})$}
\State $\bm{x}^{\prime} \gets \bm{x}$, $x_j^{\prime} \gets 1 - x_j^{\prime}$
\EndIf
\EndFor
\If{$\tilde{z}(\bm{x}^{\prime}) < \tilde{z}(\bm{x})$}
\State $\bm{x} \gets \bm{x}^{\prime}$
\State \textbf{goto START}
\EndIf
\For{$j_1 \in X$ in ascending order of $\Delta \tilde{z}_{j_1}^{\downarrow}(\bm{x})$} \Comment{Search $\textnormal{NB}_2(\bm{x})$}
\If{$L[j_1] = \emptyset$} \Comment{Generate $L[j_1]$}
\State Generate $L[j_1]$
\EndIf
\State $\bm{x}^{\prime} \gets \bm{x}$, $\Delta \tilde{z}^{\ast} \gets \infty$
\For{$j_2 \in (N \setminus X) \cap L[j_1]$} \Comment{Search $\textnormal{NB}_2^{(j_1)}(\bm{x})$}
\If{$\Delta \tilde{z}_{j_1,j_2}(\bm{x}) < \Delta \tilde{z}^{\ast}$} \Comment{Update $\pi(j_1)$}
\State $\Delta \tilde{z}^{\ast} \gets \Delta \tilde{z}_{j_1,j_2}(\bm{x})$, $\pi(j_1) \gets j_2$
\EndIf
\If{$\tilde{z}(\bm{x}) + \Delta \tilde{z}_{j_1,j_2}(\bm{x}) < \tilde{z}(\bm{x}^{\prime})$}
\State $\bm{x}^{\prime} \gets \bm{x}$, $x_{j_1}^{\prime} \gets 0$, $x_{j_2}^{\prime} \gets 1$
\EndIf
\EndFor
\If{$\tilde{z}(\bm{x}^{\prime}) < \tilde{z}(\bm{x})$}
\State $\bm{x} \gets \bm{x}^{\prime}$
\State \textbf{goto START}
\EndIf
\EndFor
\For{$j_1 \in X$ in ascending order of $\Delta \tilde{z}_{j_1,\pi(j_1)}(\bm{x})$} \Comment{Search $\textnormal{NB}_4(\bm{x})$}
\State $\bm{x}^{\prime} \gets \bm{x}$
\For{$j_3 \in X \cap L[\pi(j_1)]$ satisfying $j_3 \not = j_1$ and $\pi(j_1) \not= \pi(j_3)$}
\If{$\tilde{z}(\bm{x}) + \Delta \tilde{z}_{j_1,\pi(j_1),j_3,\pi(j_3)}(\bm{x}) < \tilde{z}(\bm{x}^{\prime})$}
\State $\bm{x}^{\prime} \gets \bm{x}$, $x_{j_1}^{\prime} \gets 0$, $x_{\pi(j_1)}^{\prime} \gets 1$, $x_{j_3}^{\prime} \gets 0$, $x_{\pi(j_3)}^{\prime} \gets 1$
\EndIf
\EndFor
\If{$\tilde{z}(\bm{x}^{\prime}) < \tilde{z}(\bm{x})$}
\State $\bm{x} \gets \bm{x}^{\prime}$
\State \textbf{goto START}
\EndIf
\EndFor
\end{algorithmic}
\end{spacing}
}
\end{algorithm}

\section{Adaptive control of penalty weights\label{sec:weight-control}}
In our algorithm, solutions are evaluated by the alternative evaluation function $\tilde{z}(\bm{x})$ in which the fixed penalty weight vectors $\bm{w}^+$ and $\bm{w}^-$ in the original evaluation function $z(\bm{x})$ has been replaced by $\widetilde{\bm{w}}^+$ and $\widetilde{\bm{w}}^-$, respectively, and the values of $\widetilde{w}_i^+$ for $i \in M_L \cup M_E$ and $\widetilde{w}_i^-$ for $i \in M_G \cup M_E$ are adaptively controlled in the search.

It is often reported that a single application of LS tends to stop at a locally optimal solution of insufficient quality when large penalty weights are used.
This is because it is often unavoidable to temporally increase the values of some violations $y_i^+$ and $y_i^-$ in order to reach an even better solution from a good solution through a sequence of neighborhood operations, and large penalty weights thus prevent LS from moving between such solutions.
To overcome this, we incorporate an adaptive adjustment mechanism for determining appropriate values of penalty weights $\widetilde{w}_i^+$ for $i \in M_L \cup M_E$ and $\widetilde{w}_i^-$ for $i \in M_G \cup M_E$ \citep{NonobeK2001,YagiuraM2006,UmetaniS2013}.
That is, LS is applied iteratively while updating the values of the penalty weights $\widetilde{w}_i^+$ for $i \in M_L \cup M_E$ and $\widetilde{w}_i^-$ for $i \in M_G \cup M_E$ after each call to LS.
We call this sequence of calls to LS the Weighting Local Search (WLS) according to \citep{SelmanB1993,ThorntonJ2005}.
This strategy is also referred as the \emph{breakout algorithm} \citep{MorrisP1993} and the \emph{dynamic local search} \citep{HutterF2002} in the literature.

Let $\bm{x}$ be the solution at which the previous local search stops.
WLS resumes LS from $\bm{x}$ after updating the penalty weight vectors $\widetilde{\bm{w}}^+$ and $\widetilde{\bm{w}}^-$.
We assume that the original penalty weights $w_i^+$ for $i \in M_L \cup M_E$ and $w_i^-$ for $i \in M_G \cup M_E$ are sufficiently large.
Starting from the original penalty weight vectors $(\widetilde{\bm{w}}^+, \widetilde{\bm{w}}^-) \gets (\bm{w}^+, \bm{w}^-)$, the penalty weight vectors $\widetilde{\bm{w}}^+$ and $\widetilde{\bm{w}}^-$ are updated as follows.
Let $\bm{x}^{\ast}$ be the best feasible solution obtained so far for the original formulation (\ref{eq:bip}).
We modify 4-FNLS to keep the best feasible solution $\bm{x}^{\ast}$, and update it whenever an improved feasible solution is found.
If $\tilde{z}(\bm{x}) \ge z(\bm{x}^{\ast})$ holds, WLS uniformly decreases the penalty weights by $(\widetilde{\bm{w}}^+, \widetilde{\bm{w}}^-) \gets \beta (\widetilde{\bm{w}}^+, \widetilde{\bm{w}}^-)$, where $0 < \beta < 1$ is a program parameter that is adaptively computed so that the new value of $\Delta \tilde{z}_j^{\downarrow}(\bm{x})$ becomes negative for 10\% of variables $x_j$ for $j \in X$.
Otherwise, WLS increases the penalty weights by
\begin{equation}
\begin{array}{ll}
\label{eq:weighting}
\widetilde{w}_i^+ \gets \displaystyle \widetilde{w}_i^+ + \frac{z(\bm{x}^{\ast}) - \tilde{z}(\bm{x})}{\sum_{l \in M} (y_l^{+^2} + y_l^{-^2})} \; y_i^+, & i \in M_L \cup M_E,\\
\widetilde{w}_i^- \gets \displaystyle \widetilde{w}_i^- + \frac{z(\bm{x}^{\ast}) - \tilde{z}(\bm{x})}{\sum_{l \in M} (y_l^{+^2} + y_l^{-^2})} \; y_i^-, & i \in M_G \cup M_E.
\end{array}
\end{equation}
WLS iteratively applies LS, updating the penalty weight vectors $\widetilde{\bm{w}}^+$ and $\widetilde{\bm{w}}^-$ after each call to LS until the time limit is reached.
WLS is formally described as Algorithm~\ref{alg:wls}.
Note that we set the initial solution to $\bm{x} = \bm{0}$ in practice.

\begin{algorithm}
\caption{WLS($\bm{x}$)\label{alg:wls}}
{\small
\begin{spacing}{1.0}
\begin{algorithmic}[1]
\Require An initial solution $\bm{x}$.
\Ensure The best feasible solution $\bm{x}^{\ast}$.
\Statex
\State $\tilde{\bm{x}} \gets \bm{x}$, $z^{\ast} \gets \infty$, $(\widetilde{\bm{w}}^+, \widetilde{\bm{w}}^-) \gets (\bm{w}^+, \bm{w}^-)$
\Repeat
\State $(\tilde{\bm{x}}, \bm{x}^{\prime}) \gets \textnormal{4-FNLS}(\tilde{\bm{x}}, \widetilde{\bm{w}}^+, \widetilde{\bm{w}}^-)$
\State \Comment{$\bm{x}^{\prime}$ is the best feasible solution obtained in $\textnormal{4-FNLS}(\tilde{\bm{x}}, \widetilde{\bm{w}}^+, \widetilde{\bm{w}}^-)$}
\If{$z(\bm{x}^{\prime}) < z^{\ast}$}
\State $\bm{x}^{\ast} \gets \bm{x}^{\prime}$, $z^{\ast} \gets z(\bm{x}^{\prime})$
\EndIf
\If{$\tilde{z}(\tilde{\bm{x}}) \ge z^{\ast}$}
\State Decrease the penalty weights by $(\widetilde{\bm{w}}^+, \widetilde{\bm{w}}^-) \gets \beta (\widetilde{\bm{w}}^+, \widetilde{\bm{w}}^-)$
\Else
\State Increase the penalty weights by (\ref{eq:weighting})
\EndIf
\Until{The time limit is reached}
\end{algorithmic}
\end{spacing}
}
\end{algorithm}

\section{Computational results\label{sec:result}}
We report computational results of our algorithm for the SCP instances from \citep{BeasleyJE1990,UmetaniS2013} and the SPP instances from \citep{BorndorferR1998,ChuPC1998,KochT2011}.
Tables~\ref{tb:scp-instance} and \ref{tb:spp-instance} summarize the information about the SCP and SPP instances, respectively.
The first column shows the name of the group (or the instance), and the number in parentheses shows the number of instances in the group.
In the subsequent part of this paper, we show the average value in each cell of the tables for the instances of the group.
The detailed computational results are in the online supplement.
The second column ``$z_{\scalebox{0.5}{LP}}$'' shows the optimal values of the LP relaxation problems.
The third column ``$z_{\scalebox{0.5}{best}}$'' shows the best upper bounds among all algorithms in this paper.
The fourth and sixth columns ``\#cst.'' show the number of constraints, and the fifth and seventh columns ``\#vars.'' show the number of variables.
Since several preprocessing techniques that often reduce the size of instances by removing redundant rows and columns are known \citep{BorndorferR1998}, all algorithms are tested on the presolved instances.
The instances marked ``$\star$'' are hard instances that cannot be solved optimally within at least 1~h by the tested Mixed Integer Programming (MIP) solvers (i.e., CPLEX12.6, Gurobi5.6.3 and SCIP3.1).

\begin{table}[t]
\caption{Benchmark instances for SCP\label{tb:scp-instance}}
\centering
\smallskip
{\tabcolsep=0.5em
{\tiny
\begin{tabular}{lrrrrcrrr} \hline
& & & \multicolumn{2}{c}{original} & & \multicolumn{2}{c}{presolved} & \\ \cline{4-5} \cline{7-8}
instance & \multicolumn{1}{c}{$z_{\scalebox{0.5}{LP}}$} & \multicolumn{1}{c}{$z_{\scalebox{0.5}{best}}$} & \multicolumn{1}{c}{\#cst.} & \multicolumn{1}{c}{\#var.} & & \multicolumn{1}{c}{\#cst.} & \multicolumn{1}{c}{\#var.} & \multicolumn{1}{c}{time limit} \\ \hline
$\star$G.1--5 (5) & 149.48 & 166.4 & 1000.0 & 10000.0 & & 1000.0 & 10000.0 & 600~s \\
$\star$H.1--5 (5) & 45.67 & 59.6 & 1000.0 & 10000.0 & & 1000.0 & 10000.0 & 600~s \\
$\star$I.1--5 (5) & 138.97 & 158.0 & 1000.0 & 50000.0 & & 1000.0 & 49981.0 & 1200~s \\
$\star$J.1--5 (5) & 104.78 & 129.0 & 1000.0 & 100000.0 & & 1000.0 & 99944.8 & 1200~s \\
$\star$K.1--5 (5) & 276.67 & 313.2 & 2000.0 & 100000.0 & & 2000.0 & 99971.0 & 1800~s \\
$\star$L.1--5 (5) & 209.34 & 258.0 & 2000.0 & 200000.0 & & 2000.0 & 199927.6 & 1800~s \\
$\star$M.1--5 (5) & 415.78 & 549.8 & 5000.0 & 500000.0 & & 5000.0 & 499988.0 & 3600~s \\
$\star$N.1--5 (5) & 348.93 & 503.8 & 5000.0 & 1000000.0 & & 5000.0 & 999993.2 & 3600~s \\
RAIL507 & 172.15 & $\ast$174 & 507 & 63009 & & 440 & 20700 & 600~s \\
RAIL516 & 182.00 & $\ast$182 & 516 & 47311 & & 403 & 37832 & 600~s \\
RAIL582 & 209.71 & $\ast$211 & 582 & 55515 & & 544 & 27427 & 600~s \\
RAIL2536 & 688.40 & $\ast$689 & 2536 & 1081841 & & 2001 & 480597 & 3600~s \\
$\star$RAIL2586 & 935.92 & 947 & 2586 & 920683 & & 2239 & 408724 & 3600~s \\
$\star$RAIL4284 & 1054.05 & 1064 & 4284 & 1092610 & & 3633 & 607884 & 3600~s \\
$\star$RAIL4872 & 1509.64 & 1530 & 4872 & 968672 & & 4207 & 482500 & 3600~s \\ \hline
\end{tabular}
}
}
\end{table}

\begin{table}[t]
\caption{Benchmark instances for SPP\label{tb:spp-instance}}
\centering
\smallskip
{\tabcolsep=0.5em
{\tiny
\begin{tabular}{lrrrrcrrr} \hline
& & & \multicolumn{2}{c}{original} & & \multicolumn{2}{c}{presolved} & \\ \cline{4-5} \cline{7-8}
instance & \multicolumn{1}{c}{$z_{\scalebox{0.5}{LP}}$} & \multicolumn{1}{c}{$z_{\scalebox{0.5}{best}}$} & \multicolumn{1}{c}{\#cst.} & \multicolumn{1}{c}{\#var.} & & \multicolumn{1}{c}{\#cst.} & \multicolumn{1}{c}{\#var.} & \multicolumn{1}{c}{time limit} \\ \hline
aa01--06 (6) & 40372.75 & $\ast$40588.83 & 675.3 & 7587.3 & & 478.7 & 6092.7 & 600~s \\
us01--04 (4) & 9749.44 & $\ast$9798.25 & 121.3 & 295085.0 & & 65.5 & 85772.5 & 600~s \\
v0415--0421 (7) & 2385764.17 & $\ast$2393303.71 & 1479.3 & 30341.6 & & 263.9 & 7277.0 & 600~s \\
v1616--1622 (7) & 1021288.76 & $\ast$1025552.43 & 1375.7 & 83986.7 & & 1171.9 & 51136.7 & 600~s\\
t0415--0421 (7) & 5199083.74 & 5453475.71 & 1479.3 & 7304.3 & & 820.7 & 2617.4 & 600~s \\
$\star$t1716--1722 (7) & 121445.76 & 157516.29 & 475.7 & 58981.3 & & 475.7 & 13193.6 & 3600~s \\
$\star$ds & 57.23 & 187.47 & 656 & 67732 & & 656 & 67730 & 3600~s \\
$\star$ds-big & 86.82 & 731.69 & 1042 & 174997 & & 1042 & 173026 & 3600~s \\
$\star$ivu06-big & 135.43 & 166.02 & 1177 & 2277736 & & 1177 & 2197774 & 3600~s \\
$\star$ivu59 & 884.46 & 1878.83 & 3436 & 2569996 & & 3413 & 2565083 & 3600~s \\ \hline
\end{tabular}
}
}
\end{table}

We compare the results of our algorithm with those of its variations.
Our algorithm was implemented in C language and tested on a MacBook Pro laptop computer with a 2.7~GHz Intel Core i7 processor and 16~GB memory. 
All variations of our algorithm were run on a single thread under MacOS10.12 operating system with time limits as shown in Tables~\ref{tb:scp-instance} and \ref{tb:spp-instance}.

Tables~\ref{tb:scp-result1} and \ref{tb:spp-result1} show the relative gap $\frac{z(\bm{x}) - z_{\scalebox{0.5}{best}}}{z(\bm{x})} \times 100$ (\%) of the best feasible solutions for the original formulation (\ref{eq:bip}) achieved by the variations of our algorithm for the SCP and SPP instances, respectively.
We note that all variations of our algorithm found feasible solutions for all SCP instances.
The second column ``no-list'' shows the results of our algorithm without the neighbor list, and the third column ``no-inc'' shows the results of our algorithm without the improved incremental evaluation (i.e., only applying the standard incremental evaluation in Section~\ref{sec:eval}).
The fourth column ``2-FNLS'' shows the results of our algorithm without the 4-flip neighborhood search (i.e., only applying 2-FNLS in Section~\ref{sec:variable-association}).
In Table~\ref{tb:spp-result1}, the number in parentheses shows the number of instances for which the algorithm obtained at least one feasible solution within the time limit, and the relative gap shows the average value for them.

\begin{table}[t]
\caption{Computational results of variations of the proposed algorithm for SCP instances\label{tb:scp-result1}}
\centering
\smallskip
{\tabcolsep=0.5em
{\tiny
\begin{tabular}{lrrrr} \hline
instance & \multicolumn{1}{c}{no-list} & \multicolumn{1}{c}{no-inc} & \multicolumn{1}{c}{2-FNLS} & \multicolumn{1}{c}{proposed} \\ \hline
$\star$G.1--5 (5) & 0.00\% & 0.12\% & 0.00\% & 0.00\% \\
$\star$H.1--5 (5) & 0.31\% & 0.31\% & 0.31\% & 0.00\% \\
$\star$I.1--5 (5) & 1.24\% & 0.86\% & 0.50\% & 0.50\% \\
$\star$J.1--5 (5) & 2.42\% & 1.67\% & 1.68\% & 1.53\% \\
$\star$K.1--5 (5) & 2.12\% & 1.69\% & 1.32\% & 1.26\% \\
$\star$L.1--5 (5) & 3.44\% & 3.51\% & 2.35\% & 2.05\% \\
$\star$M.1--5 (5) & 10.97\% & 8.33\% & 2.79\% & 2.65\% \\
$\star$N.1--5 (5) & 19.11\% & 22.06\% & 4.76\% & 5.47\% \\
RAIL507 & 0.00\% & 0.57\% & 0.00\% & 0.00\% \\
RAIL516 & 0.00\% & 0.00\% & 0.00\% & 0.00\% \\
RAIL582 & 0.47\% & 0.47\% & 0.47\% & 0.00\% \\
RAIL2536 & 2.68\% & 2.27\% & 1.29\% & 0.72\% \\
$\star$RAIL2586 & 2.57\% & 2.87\% & 2.27\% & 1.56\% \\
$\star$RAIL4284 & 5.42\% & 5.17\% & 2.74\% & 2.12\% \\
$\star$RAIL4872 & 4.43\% & 3.47\% & 2.36\% & 1.80\% \\ \hline
avg. (all) & 4.55\% & 4.42\% & 1.65\% & 1.56\% \\
avg. (with stars) & 4.89\% & 4.75\% & 1.77\% & 1.69\% \\ \hline
\end{tabular}
}
}
\end{table}

\begin{table}[t]
\caption{Computational results of variations of the proposed algorithm for SPP instances\label{tb:spp-result1}}
\centering
\smallskip
{\tabcolsep=0.5em
{\tiny
\begin{tabular}{lrrrr} \hline
instance & \multicolumn{1}{c}{no-list} & \multicolumn{1}{c}{no-inc} & \multicolumn{1}{c}{2-FNLS} & \multicolumn{1}{c}{proposed} \\ \hline
aa01--06 (6) & 2.33\%(6) & 2.26\%(6) & 2.07\%(6) & 1.60\%(6)  \\
us01--04 (4) & 0.04\%(4) & 1.16\%(4) & 0.63\%(4) & 0.04\%(4)   \\
v0415--0421 (7) & 0.00\%(7) & 0.00\%(7) & 0.00\%(7) & 0.00\%(7)   \\
v1616--1622 (7) & 0.62\%(7) & 0.17\%(7) & 0.09\%(7) & 0.09\%(7) \\
t0415--0421 (7) & 1.46\%(5) & 1.30\%(6) & 0.29\%(7) & 0.92\%(6)  \\
$\star$t1716--1722 (7) & 5.46\%(7) & 4.33\%(7) & 5.71\%(7) & 2.45\%(7)   \\
$\star$ds & 36.03\% & 33.80\% & 24.13\% & 0.00\% \\
$\star$ds-big & 29.11\% & 0.00\% & 40.75\% & 0.00\%   \\
$\star$ivu06-big & 5.31\% & 3.83\% & 2.25\% & 0.00\%   \\
$\star$ivu59 & 15.75\% & 11.39\% & 16.01\% & 0.00\%   \\ \hline
avg. (all) & 3.76\%(40/42) & 2.60\%(41/42) & 3.35\%(42/42) & 0.81\%(41/42) \\
avg. (with stars) & 10.37\%(12/13) & 6.61\%(12/13) & 9.48\%(13/13) & 1.43\%(12/13) \\ \hline
\end{tabular}
}
}
\end{table}

Tables~\ref{tb:scp-result2} and \ref{tb:spp-result2} show the computational efficiency of variations of our algorithm with respect to the number of calls to 4-FNLS (and 2-FNLS in the fourth column ``2-FNLS''), where the bottom rows show average factors normalized so that that of our algorithm is set to one.

\begin{table}[t]
\caption{The number of calls to 4-FNLS of variations of the proposed algorithm for SCP instances\label{tb:scp-result2}}
\centering
\smallskip
{\tabcolsep=0.5em
{\tiny
\begin{tabular}{lrrrr} \hline
instance & \multicolumn{1}{c}{no-list} & \multicolumn{1}{c}{no-inc} & \multicolumn{1}{c}{2-FNLS} & \multicolumn{1}{c}{proposed} \\ \hline
$\star$G.1--5 (5) & 1528.2 & 1094.2 & 5755.8 & 3001.4\\
$\star$H.1--5 (5) & 1040.8 & 434.4 & 2191.6 & 1248.2 \\
$\star$I.1--5 (5) & 845.2 & 750.6 & 5072.0 & 2189.8 \\
$\star$J.1--5 (5) & 494.0 & 454.6 & 2580.4 & 1057.8 \\
$\star$K.1--5 (5) & 364.0 & 444.8 & 2828.8 & 1445.8 \\
$\star$L.1--5 (5) & 262.4 & 273.8 & 1862.6 & 809.0 \\
$\star$M.1--5 (5) & 103.0 & 117.2 & 1412.8 & 535.0 \\
$\star$N.1--5 (5) & 78.8 & 66.0 & 929.4 & 290.8 \\
RAIL507 & 4498 & 6482 & 38195 & 23692 \\
RAIL516 & 2004 & 3345 & 16915 & 12350 \\
RAIL582 & 3079 & 4037 & 22123 & 15140 \\
RAIL2536 & 429 & 359 & 4045 & 3206 \\
$\star$RAIL2586 & 416 & 399 & 5398 & 3028 \\
$\star$RAIL4284 & 182 & 195 & 2665 & 1683 \\
$\star$RAIL4872 & 197 & 229 & 4214 & 2374 \\ \hline
avg. factor & 0.37 & 0.30 & 2.21 & 1.00  \\ \hline
\end{tabular}
}
}
\end{table}

\begin{table}[t]
\caption{The number of calls to 4-FNLS of variations of the proposed algorithm for SPP instances\label{tb:spp-result2}}
\centering
\smallskip
{\tabcolsep=0.5em
{\tiny
\begin{tabular}{lrrrr} \hline
instance & \multicolumn{1}{c}{no-list} & \multicolumn{1}{c}{no-inc} & \multicolumn{1}{c}{2-FNLS} & \multicolumn{1}{c}{proposed} \\ \hline
aa01--06 (6) & 10152.2 & 20544.8 & 93694.3 & 56584.8 \\
us01--04 (4) & 35418.5 & 36491.75 & 92397.5 & 86031.0 \\
v0415--0421 (7) & 1516414.7 & 1783742.3 & 4274501.6 & 2749825.9 \\
v1616--1622 (7) & 1197.7 & 4029.1 & 17626.9 & 10048.7 \\
t0415--0421 (7) & 9412.6 & 19885.4 & 145916.4 & 53751.7 \\
$\star$t1716--1722 (7) & 32689.3 & 60992.9 & 281283.4 & 174405.1 \\
$\star$ds & 3109 & 3463 & 14142 & 11414 \\
$\star$ds-big & 852 & 1019 & 4370 & 3294 \\
$\star$ivu06-big & 236 & 305 & 1293 & 1060 \\
$\star$ivu59 & 141 & 208 & 620 & 438 \\ \hline
avg. factor & 0.20 & 0.38 & 1.70 & 1.00 \\ \hline
\end{tabular}
}
}
\end{table}

From these results, we observe that the proposed method improves the computational efficiency of the local search algorithm for the SCP and SPP instances in comparison with its variations, and our algorithm attains good performance even when the size of the neighbor list is considerably small.
We also observe that the 4-flip neighborhood search substantially improves the performance of our algorithm even though there are fewer calls to 4-FNLS compared to 2-FNLS.

Tables~\ref{tb:scp-result3} and \ref{tb:spp-result3} show the ratio $\frac{\#\textnormal{generated rows}}{|N|} \times 100$ (\%) of generated rows in the neighbor list.
We observe that our algorithm achieves good performance while generating only a small part of the neighbor list for the large-scale instances.

\begin{table}[t]
\caption{The ratio of generated rows in the neighbor list for SCP instances\label{tb:scp-result3}}
\centering
\smallskip
{\tabcolsep=0.5em
{\tiny
\begin{tabular}{lrrrrr} \hline
instance & \multicolumn{1}{c}{1 min} & \multicolumn{1}{c}{10 min} & \multicolumn{1}{c}{20 min} & \multicolumn{1}{c}{30 min} & \multicolumn{1}{c}{1 h}\\ \hline
$\star$G.1--5 (5) & 3.58\% & 3.95\% \\
$\star$H.1--5 (5) & 2.12\% & 2.43\% \\
$\star$I.1--5 (5) & 1.39\% & 1.63\% & 1.71\% \\
$\star$J.1--5 (5) & 0.82\% & 1.10\% & 1.16\% \\
$\star$K.1--5 (5) & 1.20\% & 1.49\% & 1.57\% & 1.61\% \\
$\star$L.1--5 (5) & 0.57\% & 0.98\% & 1.06\% & 1.10\%\\
$\star$M.1--5 (5) & 0.20\% & 0.58\% & 0.73\% & 0.81\% & 0.93\% \\
$\star$N.1--5 (5) & 0.01\% & 0.21\% & 0.29\% & 0.35\% & 0.49\% \\
RAIL507 & 13.41\% & 22.02\% \\
RAIL516 & 4.73\% & 8.94\% \\
RAIL582 & 8.19\% & 10.66\% \\
RAIL2536 & 0.30\% & 1.36\% & 1.87\% & 2.17\% & 2.73\% \\
$\star$RAIL2586 & 0.44\% & 1.67\% & 2.19\% & 2.54\% & 3.28\%  \\
$\star$RAIL4284 & 0.20\% & 1.01\% & 1.43\% & 1.73\% & 2.29\%  \\
$\star$RAIL4872 & 0.37\% & 1.56\% & 2.11\% & 2.46\% & 3.10\%  \\ \hline
\end{tabular}
}
}
\end{table}

\begin{table}[t]
\caption{The ratio of generated rows in the neighbor list for SPP instances\label{tb:spp-result3}}
\centering
\smallskip
{\tabcolsep=0.5em
{\tiny
\begin{tabular}{lrrrr} \hline
instance & \multicolumn{1}{c}{1 min} & \multicolumn{1}{c}{10 min} & \multicolumn{1}{c}{30 min} & \multicolumn{1}{c}{1 h}\\ \hline
aa01--06 (6) & 40.47\% & 49.17\% \\
us01--04 (4) & 3.93\% & 5.17\% \\
v0415--0421 (7) & 31.44\% & 31.64\% \\
v1616--1622 (7) & 6.55\% & 7.42\% \\
t0415--0421 (7) & 83.47\% & 90.00\% \\
$\star$t1716--1722 (7) & 61.00\% & 94.38\% & 97.12\% & 97.98\% \\
$\star$ds & 2.29\% & 12.63\% & 27.11\% & 40.05\% \\
$\star$ds-big & 0.21\% & 2.06\% & 4.96\% & 8.10\% \\
$\star$ivu06-big & 0.01\% & 0.07\% & 0.23\% & 0.45\% \\
$\star$ivu59 & 0.01\% & 0.05\% & 0.11\% & 0.16\% \\ \hline
\end{tabular}
}
}
\end{table}

We compare the result of our algorithm with those of general purpose solvers, i.e., the latest MIP solvers called \citet{CPLEX12.6}, \citet{Gurobi5.6.3} and SCIP3.1 \citep{AchterbergT2009} and a local search solver for BIP (including nonlinear constraints and objective functions) called LocalSolver3.1 \citep{BenoistT2011}.
LocalSolver3.1 is a simulated annealing based on ejection chain moves specialized for maintaining the feasibility of Boolean constraints and an efficient incremental evaluation using a directed acyclic graph.
LocalSolver3.1 is not the latest version, but it performs better than the latest version (LocalSolver4.5) for the SCP and SPP instances.

We note the following issues in comparing the performance of our algorithm with that of the latest MIP solvers.
Due to the lack of pruning mechanism, it is inherently difficult to find optimal solutions by local search algorithms even for instances having a small gap between the lower and upper bounds of the optimal values, while the latest MIP solvers find optimal solutions quickly by the branch-and-cut procedure.
Of course, local search algorithms often obtain good upper bounds close to the optimal values for the instances.
On the other hand, the latest MIP solvers often prefer running primal heuristics rather than the branch-and-cut procedure for instances having a large gap between the lower and upper bounds of the optimal values.
Indeed, the latest MIP solvers include dozens of primal heuristics and spend much computation time on finding good feasible solutions \citep{LodiA2013}, e.g., SCIP3.1 reported that it spent 1644.73~s out of 3600~s running primal heuristics for solving the ``ds'' instance, while it spent 837.47~s for solving LP relaxation problems.

We also compare our algorithm with a 3-flip local search algorithm specially tailored for SCP developed by \citet{YagiuraM2006} (denoted by YKI).
In order to handle large-scale SCP instances, many heuristic algorithms have introduced the pricing method that reduces the number of variables to be considered by using LP and/or Lagrangian relaxation.
YKI introduced a pricing method based on Lagrangian relaxation that substantially reduces the number of variables to be considered to 1.05\% from the original SCP instances on average.
Hence, in addition to the original SCP instances, we also tested all algorithms for reduced SCP instances by another pricing method based on LP relaxation \citep{UmetaniS2007}.
Table~\ref{tb:scp-reduce-instance} summarize the reduced SCP instances, where we applied preprocessing to the reduced SCP instances as well as the original SCP instances.
The eighth column ``\#free var.'' shows the number of variables to be considered in YKI for the original SCP instances.
We note that it turned out that many of reduced SPP instances by the pricing method \citep{UmetaniS2007} were infeasible because the equality constraints of SPP often prevent solutions from containing highly evaluated variables together and make the pricing method less effective.

\begin{table}[t]
\caption{Reduced benchmark instances for SCP\label{tb:scp-reduce-instance}}
\centering
\smallskip
{\tabcolsep=0.5em
{\tiny
\begin{tabular}{lrrrrcrrcrr} \hline
& & & \multicolumn{2}{c}{reduced} & & \multicolumn{2}{c}{presolved} & & \multicolumn{1}{c}{Yagiura et~al.} \\ \cline{4-5} \cline{7-8} \cline{10-10}
instance & \multicolumn{1}{c}{$z_{\scalebox{0.5}{LP}}$} & \multicolumn{1}{c}{$z_{\scalebox{0.5}{best}}$} & \multicolumn{1}{c}{\#cst.} & \multicolumn{1}{c}{\#var.} & & \multicolumn{1}{c}{\#cst.} & \multicolumn{1}{c}{\#var.} & & \multicolumn{1}{c}{\#free var.} & \multicolumn{1}{c}{time limit} \\ \hline
$\star$G.1--5 (5) & 149.48 & 166.4 & 1000.0 & 441.8 & & 1000.0 & 441.8 & & 339.0 & 600~s \\
$\star$H.1--5 (5) & 45.67 & 59.8 & 1000.0 & 236.2 & & 1000.0 & 236.2 & & 174.6 & 600~s \\
$\star$I.1--5 (5) & 138.97 & 158.4 & 1000.0 & 721.6 & & 1000.0 & 721.6 & & 479.4 & 1200~s \\
$\star$J.1--5 (5) & 104.78 & 129.4 & 1000.0 & 703.6 & & 1000.0 & 703.6 & & 429.6 & 1200~s \\
$\star$K.1--5 (5) & 276.67 & 313.8 & 2000.0 & 1434.2 & & 2000.0 & 1434.2 & & 959.4 & 1800~s \\
$\star$L.1--5 (5) & 209.34 & 259.0 & 2000.0 & 1421.0 & & 2000.0 & 1421.0 & & 856.4 & 1800~s \\
$\star$M.1--5 (5) & 415.78 & 551.6 & 5000.0 & 3245.2 & & 5000.0 & 3245.2 & & 1836.4 & 3600~s \\
$\star$N.1--5 (5) & 348.93 & 505.0 & 5000.0 & 4471.4 & & 5000.0 & 4471.4 & & 1660.6 & 3600~s \\
RAIL507 & 172.15 & $\ast$174 & 507 & 2649 & & 402 & 1019 & & 394 & 600~s \\
RAIL516 & 183.00 & $\ast$183 & 516 & 3788 & & 350 & 3088 & & 492 & 600~s \\
RAIL582 & 209.71 & $\ast$211 & 582 & 2091 & & 491 & 1493 & & 513 & 600~s \\
RAIL2536 & 688.68 & $\ast$691 & 2536 & 13746 & & 1391 & 5782 & & 1598 & 3600~s \\
$\star$RAIL2586 & 935.92 & 948 & 2586 & 13349 & & 2083 & 7377 & & 2089 & 3600~s \\
$\star$RAIL4284 & 1054.05 & 1066 & 4284 & 21728 & & 3189 & 14565 & & 2639 & 3600~s \\
$\star$RAIL4872 & 1510.87 & 1532 & 4872 & 21329 & & 3577 & 11404 & & 3650 & 3600~s \\ \hline
\end{tabular}
}
}
\end{table}

Tables~\ref{tb:scp-result4} and \ref{tb:scp-reduce-result4} show the relative gap of the best feasible solutions for the original formulation (\ref{eq:bip}) achieved by the algorithms for the original and reduced SCP instances, respectively.
We note that all algorithms found feasible solutions for all SCP instances.
All algorithms were tested on a MacBook Pro laptop computer with a 2.7~GHz Intel Core i7 processor and 16~GB memory, and were run on a single thread under MacOS10.12 operating system with time limits as shown in Tables~\ref{tb:scp-instance} and \ref{tb:scp-reduce-instance}.
That is, we tested all algorithms under the same amount of available computational resources for fair comparison.

\begin{table}[t]
\caption{Computational results of the latest solvers and the proposed algorithm for SCP instances\label{tb:scp-result4}}
\centering
\smallskip
{\tabcolsep=0.5em
{\tiny
\begin{tabular}{lrrrrrr} \hline
instance & \multicolumn{1}{c}{CPLEX12.6} & \multicolumn{1}{c}{Gurobi5.6.3} & \multicolumn{1}{c}{SCIP3.1} & \multicolumn{1}{c}{LocalSolver3.1} & \multicolumn{1}{c}{Yagiura et~al.} & \multicolumn{1}{c}{proposed} \\ \hline
$\star$G.1--5 (5) & 0.37\% & 0.49\% & 0.24\% & 45.80\% & 0.00\% & 0.00\% \\
$\star$H.1--5 (5) & 1.92\% & 2.28\% & 1.93\% & 61.54\% & 0.00\% & 0.00\% \\
$\star$I.1--5 (5) & 2.81\% & 2.72\% & 1.85\% & 41.38\% & 0.00\% & 0.50\% \\
$\star$J.1--5 (5) & 8.37\% & 4.30\% & 3.59\% & 58.40\% & 0.00\% & 1.53\% \\
$\star$K.1--5 (5) & 4.77\% & 4.38\% & 2.55\% & 51.22\% & 0.00\% & 1.26\% \\
$\star$L.1--5 (5) & 9.57\% & 8.44\% & 3.52\% & 57.79\% & 0.00\% & 2.05\% \\
$\star$M.1--5 (5) & 18.43\% & 10.10\% & 30.71\% & 71.08\% & 0.00\% & 2.65\% \\
$\star$N.1--5 (5) & 33.13\% & 12.49\% & 42.32\% & 75.63\% & 0.00\% & 5.47\% \\
RAIL507 & 0.00\% & 0.00\% & 0.00\% & 5.43\% & 0.00\% & 0.00\% \\
RAIL516 & 0.00\% & 0.00\% & 0.00\% & 3.19\% & 0.00\% & 0.00\% \\
RAIL582 & 0.00\% & 0.00\% & 0.00\% & 5.80\% & 0.00\% & 0.00\% \\
RAIL2536 & 0.00\% & 0.00\% & 0.86\% & 3.50\% & 0.29\% & 0.72\% \\
$\star$RAIL2586 & 2.27\% & 2.17\% & 2.27\% & 5.39\% & 0.00\% & 1.56\% \\
$\star$RAIL4284 & 5.34\% & 1.57\% & 30.55\% & 6.50\% & 0.00\% & 2.12\% \\
$\star$RAIL4872 & 1.73\% & 1.73\% & 2.67\% & 5.61\% & 0.00\% & 1.80\% \\ \hline
avg. (all) & 8.64\% & 4.92\% & 10.00\% & 49.99\% & 0.01\% & 1.56\% \\
avg. (with stars) & 9.45\% & 5.38\% & 10.91\% & 54.22\% & 0.00\% & 1.69\% \\ \hline
\end{tabular}
}
}
\end{table}

\begin{table}[t]
\caption{Computational results of the latest solvers and the proposed algorithm for reduced SCP instances\label{tb:scp-reduce-result4}}
\centering
\smallskip
{\tabcolsep=0.5em
{\tiny
\begin{tabular}{lrrrrrr} \hline
instance & \multicolumn{1}{c}{CPLEX12.6} & \multicolumn{1}{c}{Gurobi5.6.3} & \multicolumn{1}{c}{SCIP3.1} & \multicolumn{1}{c}{LocalSolver3.1} & \multicolumn{1}{c}{Yagiura et~al.} & \multicolumn{1}{c}{proposed} \\ \hline
$\star$G.1--5 (5) & 0.60\% & 0.24\% & 0.47\% & 3.78\% & 0.00\% & 0.00\% \\
$\star$H.1--5 (5) & 1.92\% & 1.62\% & 1.62\% & 1.90\% & 0.00\% & 0.00\% \\
$\star$I.1--5 (5) & 2.59\% & 1.64\% & 2.10\% & 1.74\% & 0.00\% & 0.00\% \\
$\star$J.1--5 (5) & 4.43\% & 3.99\% & 3.42\% & 2.99\% & 0.15\% & 0.31\% \\
$\star$K.1--5 (5) & 2.84\% & 2.48\% & 2.66\% & 2.18\% & 0.00\% & 0.63\% \\
$\star$L.1--5 (5) & 4.77\% & 4.85\% & 3.13\% & 2.41\% & 0.00\% & 0.77\% \\
$\star$M.1--5 (5) & 10.82\% & 4.56\% & 31.15\% & 4.89\% & 0.00\% & 1.01\%\\
$\star$N.1--5 (5) & 15.95\% & 11.36\% & 34.17\% & 6.47\% & 0.00\% & 1.06\% \\
RAIL507 & 0.00\% & 0.00\% & 0.00\% & 0.57\% & 0.00\% & 0.00\% \\
RAIL516 & 0.00\% & 0.00\% & 0.00\% & 1.61\% & 0.00\% & 0.00\% \\
RAIL582 & 0.00\% & 0.00\% & 0.00\% & 0.47\% & 0.00\% & 0.00\% \\
RAIL2536 & 0.00\% & 0.00\% & 0.00\% & 0.29\% & 0.00\% & 0.00\% \\
$\star$RAIL2586 & 0.84\% & 0.52\% & 0.73\% & 1.66\% & 0.00\% & 0.84\% \\
$\star$RAIL4284 & 0.47\% & 0.93\% & 0.74\% & 1.66\% & 0.00\% & 0.84\% \\
$\star$RAIL4872 & 0.84\% & 1.23\% & 1.10\% & 1.73\% & 0.00\% & 0.84\% \\ \hline
avg. (all) & 4.72\% & 3.33\% & 8.43\% & 2.98\% & 0.02\% & 0.45\% \\
avg. (with stars) & 5.16\% & 3.64\% & 9.21\% & 3.18\% & 0.02\% & 0.50\% \\ \hline
\end{tabular}
}
}
\end{table}

We observe that our algorithm achieves better upper bounds than general purpose solvers (CPLEX12.6.3, Gurobi5.6.3, SCIP3.1 and LocalSolver3.1) for the original and reduced SCP instances.
We also observe that our algorithm achieves good upper bounds close to those of YKI for the reduced SCP instances.
It shows that our algorithm attains good performance close to that of the specially tailored algorithm (YKI) for SCP by combining a pricing method.
In view of these, our algorithm achieves sufficiently good upper bounds compared to the other algorithms for the SCP instances.

Table~\ref{tb:spp-result4} shows the relative gap of the best feasible solutions achieved by the algorithms for the SPP instances.
All algorithms were tested on a MacBook Pro laptop computer with a 2.7~GHz Intel Core i7 processor and 16~GB memory, and were run on a single thread under MacOS~10.12 operating system with time limits as shown in Table~\ref{tb:spp-instance}.
In Table~\ref{tb:spp-result4}, the number in parentheses shows that of instances for which the algorithm obtained at least one feasible solution within the time limit, and the relative gap shows the average value for them.

We also compare the computational results of our algorithm with those of a Lagrangian heuristic algorithm for BIP called the Wedelin's heuristic \citep{BastertO2010,WedelinD1995} and a branch-and-cut algorithm specially tailored for SPP developed by \citet{BorndorferR1998}.
The computational results for the Wedelin's heuristic are taken from \citep{BastertO2010}, where it was tested on a 1.3~GHz Sun Ultra Sparc-IIIi and run with a time limit of 600~s.
The computational results for the branch-and-cut algorithm are taken from \citep{BorndorferR1998}, where it was tested on a Sun Ultra Sparc 1 Model 170E and run with a time limit of 7200~s.

\begin{table}[t]
\caption{Computational results of the latest solvers and the proposed algorithm for SPP instances\label{tb:spp-result4}}
\centering
\smallskip
{\tabcolsep=0.5em
{\tiny
\begin{tabular}{lrrrrrrr} \hline
instance & \multicolumn{1}{c}{CPLEX12.6} & \multicolumn{1}{c}{Gurobi5.6.3} & \multicolumn{1}{c}{SCIP3.1} & \multicolumn{1}{c}{LocalSolver3.1} & \multicolumn{1}{c}{Bastert et~al.$\dagger$} & \multicolumn{1}{c}{Bornd{\"o}rfer$\ddagger$} & \multicolumn{1}{c}{proposed} \\ \hline
aa01--06 (6) & 0.00\%(6) & 0.00\%(6) & 0.00\%(6) & 13.89\%(1) & --- & 0.00\%(6) & 1.60\%(6) \\
us01--04 (4) & 0.00\%(4) & 0.00\%(4) & 0.00\%(3) & 11.26\%(2) & --- & 0.00\%(4) & 0.04\%(4) \\
v0415--0421 (7) & 0.00\%(7) & 0.00\%(7) & 0.00\%(7) & 0.05\%(7) & 0.71\%(6) & 0.00\%(7) & 0.00\%(7) \\
v1616--1622 (7) & 0.00\%(7) & 0.00\%(7) & 0.00\%(7) & 4.60\%(7) & 6.64\%(3) & 0.01\%(7) & 0.09\%(7) \\
t0415--0421 (7) & 0.66\%(7) & 0.60\%(7) & 1.61\%(6) & --- (0) & 1.30\%(5) & 1.83\%(7) & 0.92\%(6) \\
$\star$t1716--1722 (7) & 8.34\%(7) & 16.58\%(7) & 3.51\%(7) & 37.08\%(1) & 12.55\%(7) & 1.63\%(7) & 2.45\%(7) \\
$\star$ds & 8.86\% & 55.61\% & 40.53\% & 85.17\% & 8.82\% & --- & 0.00\% \\
$\star$ds-big & 62.16\% & 24.03\% & 72.01\% & 92.69\% & --- & --- & 0.00\% \\
$\star$ivu06-big & 20.86\% & 0.68\% & 17.90\% & 52.54\% & --- & --- & 0.00\% \\
$\star$ivu59 & 28.50\% & 4.36\% & 37.84\% & 48.95\% & --- & --- & 0.00\% \\ \hline
avg. (all) & 4.37\%(42/42) & 4.88\%(42/42) & 5.06\%(40/42) & 17.52\%(22/42) & 5.79\%(22/29) & 0.64\%(38/38) & 0.81\%(41/42) \\
avg. (with stars) & 14.10\%(13/13) & 15.66\%(13/13) & 15.07\%(13/13) & 63.29\%(5/13) & 10.79\%(9/10) & 1.77\%(9/9) & 1.43\%(12/13) \\ \hline
\multicolumn{8}{l}{$\dagger$ 600~s on a 1.3~GHz Sun Ultra Sparc-IIIi}\\
\multicolumn{8}{l}{$\ddagger$  7200~s on a Sun Ultra Sparc 1 Model 170E}
\end{tabular}
}
}
\end{table}

We first observe that our algorithm achieves better upper bounds than the latest MIP solvers (CPLEX12.6, Gurobi5.6.3 and SCIP3.1) for hard SPP instances and good upper bounds close to the optimal values for the others.
We next observe that our algorithm achieves better upper bounds than the general purpose heuristic solvers (LocalSolver3.1 and the Wedelin's heuristic) and good upper bounds comparable to those of the branch-and-bound algorithm \mbox{\citep{BorndorferR1998}} for many SPP instances.
In view of these, our algorithm achieves sufficiently good upper bounds compared to the other algorithms for the SPP instances, particularly for hard SPP instances.

\section{Conclusion\label{sec:conclusion}}
We present a data mining approach for reducing the search space of local search algorithms for a class of BIPs including SCP and SPP.
In this approach, we construct a $k$-nearest neighbor graph by extracting variable associations from the instance to be solved in order to identify promising pairs of flipping variables in the 2-flip neighborhood.
We also develop a 4-flip neighborhood local search algorithm that flips four variables alternately along 4-paths or 4-cycles in the $k$-nearest neighbor graph.
We incorporate an efficient incremental evaluation of solutions and an adaptive control of penalty weights into the 4-flip neighborhood local search algorithm. 
Computational results show that the proposed method improves the performance of the local search algorithm for large-scale SCP and SPP instances.

We expect that data mining approaches could also be beneficial for efficiently solving other large-scale combinatorial optimization problems, particularly for hard instances having large gaps between the lower and upper bounds of the optimal values.

\section*{Acknowledgment}
This work was supported by the Grants-in-Aid for Scientific Research (JP26282085).



\section*{References}
\begin{thebibliography}{1}
\bibitem[Achterberg(2009)]{AchterbergT2009}
Achterberg, T. (2009).
SCIP: Solving constraint integer programs.
\emph{Mathematical Programming Computation}, 1, 1--41.

\bibitem[Agarwal et~al.(1989)]{AgarwalY1989}
Agarwal, Y., Mathur, K., \& Salkin, H.~M. (1989).
A set-partitioning-based exact algorithm for the vehicle routing problem.
\emph{Networks}, 19, 731--749.

\bibitem[Atamt{\"u}rk et~al.(1995)]{AtamturkA1995}
Atamt{\"{u}}rk, A., Nemhauser, G.~L., \& Savelsbergh, M.~W.~P. (1995).
A combined Lagrangian, linear programming, and implication heuristic for large-scale set partitioning problems.
\emph{Journal of Heuristics}, 1, 247--259.

\bibitem[Balas \& Padberg(1976)]{BalasE1976}
Balas, E., \& Padberg, M.~W. (1976).
Set partitioning: A survey.
\emph{SIAM Review}, 18, 710--760.

\bibitem[Baldacci et~al.(2008)]{BaldacciR2008}
Baldacci, R., Christofides, N., \& Mingozzi, A. (2008).
An exact algorithm for the vehicle routing problem based on the set partitioning formulation with additional cuts.
\emph{Mathematical Programming}, 115, 351--385.

\bibitem[Barahona \& Anbil(2000)]{BarahonaF2000}
Barahona, F. \& Anbil, R. (2000).
The volume algorithm: Producing primal solutions with a subgradient method.
\emph{Mathematical Programming}, 87, 385--399.

\bibitem[Barnhart et~al.(1998)]{BarnhartC1998}
Barnhart, C., Johnson, E.~L., Nemhauser, G.~L., Savelsbergh, M.~W.~P., \& Vance, P.~H. (1998).
Branch-and-price: Column generation for solving huge integer programs.
\emph{Operations Research}, 46, 316--329.

\bibitem[Bastert et~al.(2010)]{BastertO2010}
Bastert, O., Hummel, B. \& de Vries, S. (2010).
A generalized Wedelin heuristic for integer programming.
\emph{INFORMS Journal on Computing}, 22, 93--107.

\bibitem[Beasley(1990)]{BeasleyJE1990}
Beasley, J.~E. (1990).
OR-Library: Distributing test problems by electronic mail.
\emph{Journal of the Operational Research Society}, 41, 1069--1072.

\bibitem[Benoist et~al.(2011)]{BenoistT2011}
Benoist, T., Estellon, B., Gardi, F., Megel, R., \& Nouioua, K. (2011).
LocalSolver 1.x: A black-box local-search solver for 0-1 programming.
\emph{4OR --- A Quarterly Journal of Operations Research}, 9, 299--316.

\bibitem[Bornd{\"o}rfer(1998)]{BorndorferR1998}
Bornd{\"o}rfer, R. (1998).
\emph{Aspects of set packing, partitioning and covering}.
Ph.~D. Dissertation, Berlin: Technischen Universit{\"a}t.

\bibitem[Boros et~al.(2000)]{BorosE2000}
Boros, E., Hammer, P.~L., Ibaraki, T., Kogan, A., Mayoraz, E., \& Muchnik, I. (2000).
An implementation of logical analysis of data.
\emph{IEEE Transactions on Knowledge and Data Engineering}, 12, 292--306.

\bibitem[Boros et~al.(2005)]{BorosE2005}
Boros, E., Ibaraki, T., Ichikawa, H., Nonobe, K., Uno, T., \& Yagiura, M. (2005).
Heuristic approaches to the capacitated square covering problem. 
\emph{Pacific Journal of Optimization}, 1, 465--490.

\bibitem[Boschetti et~al.(2008)]{BoschettiMA2008}
Boschetti, M.~A., Mingozzi, A. \& Ricciardelli, S. (2008).
A dual ascent procedure for the set partitioning problem.
\emph{Discrete Optimization}, 5, 735--747.

\bibitem[Bramel \& Simchi-Levi(1997)]{BramelJ1997}
Bramel, J., \& Simchi-Levi, D. (1997).
On the effectiveness of set covering formulations for the vehicle routing problem with time windows.
\emph{Operations Research}, 45, 295--301.

\bibitem[Caprara et~al.(1999)]{CapraraA1999}
Caprara, A., Fischetti, M., \& Toth, P. (1999).
A heuristic method or the set covering problem.
\emph{Operations Research}, 47, 730--743.

\bibitem[Caprara et~al.(2000)]{CapraraA2000}
Caprara, A., Toth, P., \& Fischetti, M. (2000).
Algorithms for the set covering problem.
\emph{Annals of Operations Research}, 98, 353--371.

\bibitem[Caserta(2007)]{CasertaM2007}
Caserta, M. (2007).
Tabu search-based metaheuristic algorithm for large-scale set covering problems.
In W.~J.~Gutjahr, R.~F.~Hartl, \& M.~Reimann (eds.),
\emph{Metaheuristics: Progress in Complex Systems Optimization} (pp.~43--63).
Berlin: Springer.

\bibitem[Ceria et~al.(1997)]{CeriaS1997}
Ceria, S., Nobili, P., Sassano, A. (1997).
Set covering problem.
In M.~Dell'Amico, F.~Maffioli \& S.~Martello (eds.): \emph{Annotated Bibliographies in Combinatorial Optimization}, (pp.~415--428).
New Jersey: John Wiley \& Sons.

\bibitem[Ceria et~al.(1998)]{CeriaS1998}
Ceria, S., Nobili, P., \& Sassano, A. (1998).
A Lagrangian-based heuristic for large-scale set covering problems.
\emph{Mathematical Programming}, 81, 215--288.

\bibitem[Chu \& Beasley(1998)]{ChuPC1998}
Chu P.~C., \& Beasley, J.~E. (1998).
Constraint handling in genetic algorithms: The set partitioning problem.
\emph{Journal of Heuristics}, 11, 323--357.

\bibitem[CPLEX12.6(2014)]{CPLEX12.6}
CPLEX Optimizer. (2014).
\url{http://www-01.ibm.com/software/commerce/optimization/cplex-optimizer/index.html} Accessed 2017.01.23.

\bibitem[Farahani et~al.(2012)]{FarahaniRZ2012}
Farahani, R.~Z., Asgari, N., Heidari, N., Hosseininia, M., \& Goh, M. (2012).
Covering problems in facility location: A review.
\emph{Computers \& Industrial Engineering}, 62, 368--407.

\bibitem[Gurobi5.6.3(2014)]{Gurobi5.6.3}
Gurobi Optimizer. (2014).
\url{http://www.gurobi.com/} Accessed 2017.01.23.

\bibitem[Hammer \& Bonates(2006)]{HammerPL2006}
Hammer, P. L., \& Bonates, T. O. (2006).
Logical analysis of data --- An overview: From combinatorial optimization to medical applications.
\emph{Annals of Operations Research}, 148, 203--225.

\bibitem[Hashimoto et~al.(2009)]{HashimotoH2009}
Hashimoto, H., Ezaki, Y., Yagiura, M., Nonobe, K., Ibaraki, T., \& L{\o}kketangen, A. (2009).
A set covering approach for the pickup and delivery problem with general constraints on each route.
\emph{Pacific Journal of Optimization}, 5, 183--200.

\bibitem[Hoffman \& Padberg(1993)]{HoffmanKL1993}
Hoffman, K.~L., \& Padberg, A. (1993).
Solving airline crew scheduling problems by branch-and-cut.
\emph{Management Science}, 39, 657--682.

\bibitem[Hutter et~al.(2002)]{HutterF2002}
Hutter, F., Tompkins, D.~A.~D., \& Hoos, H.~H. (2002).
Scaling and probabilistic smoothing: Efficient dynamic local search for SAT.
\emph{Proceedings of International Conference on Principles and Practice of Constraint Programming (CP)}, 233--248.

\bibitem[Johnson \& McGeoch(1997)]{JohnsonDS1997}
Johnson, D.~S., \& McGeoch, L.~A. (1997).
The traveling salesman problem: A case study.
In E.~Aarts, \& K.~Lenstra (eds.),
\emph{Local Search in Combinatorial Optimization} (pp.~215--310).
New Jersey: Princeton University Press.

\bibitem[Koch et~al.(2011)]{KochT2011}
Koch, T., Achterberg, T., Andersen, E., Bastert, O., Berthold, T., Bixby, R.~E., Danna, E., Gamrath, G., Gleixner, A.~M., Heinz, S, Lodi, A., Mittelmann, H., Ralphs, T., Salvagnin, D., Steffy, D.~E., \& Wolter, K. (2011).
MIPLIB2010: Mixed integer programming library version 5.
\emph{Mathematical Programming Computation}, 3, 103--163.

\bibitem[Linderoth et~al.(2001)]{LinderothJT2001}
Linderoth, J.~T., Lee, E.~K., \& Savelbergh, M.~W.~P. (2001).
A parallel, linear programming-based heuristic for large-scale set partitioning problems.
\emph{INFORMS Journal on Computing}, 13, 191--209.

\bibitem[Lodi(2013)]{LodiA2013}
Lodi, A. (2013)
The heuristic (dark) side of MIP solvers.
In E.~-G.~Talbi (ed.),
\emph{Hybrid Metaheuristics} (pp.~273--284).
Berlin: Springer.

\bibitem[Michel \& Van Hentenryck(2000)]{MichelL2000}
Michel, L. \& Van Hentenryck, P. (2000).
Localizer.
\emph{Constraints: An International Journal}, 5, 43--84.

\bibitem[Mingozzi et~al.(1999)]{MingozziA1999}
Mingozzi, A., Boschetti, M.~A., Ricciardelli, S., \& Bianco, L. (1999).
A set partitioning approach to the crew scheduling problem.
\emph{Operations Research}, 47, 873--888.

\bibitem[Morris(1993)]{MorrisP1993}
Morris, P. (1993).
The breakout method for escaping from local minima.
\emph{Proceedings of National Conference on Artificial Intelligence (AAAI)}, 40--45.

\bibitem[Nonobe \& Ibaraki(2001)]{NonobeK2001}
Nonobe, K., \& Ibaraki, T. (2001).
An improved tabu search method for the weighted constraint satisfaction problem.
\emph{INFOR}, 39, 131--151.

\bibitem[Pesant \& Gendreau(1999)]{PesantG1999}
Pesant, G., \& Gendreau, M. (1999).
A constraint programming framework for local search methods.
\emph{Journal of Heuristics}, 5, 255--279.

\bibitem[Selman \& Kautz(1993)]{SelmanB1993}
Selman, B., \& Kautz, H. (1993).
Domain-independent extensions to GSAT: Solving large structured satisfiability problems.
\emph{Proceedings of International Conference on Artificial Intelligence (IJCAI)}, 290--295.

\bibitem[Shaw et~al.(2002)]{ShawP2002}
Shaw, P., Backer, B.~D., \& Furnon, V. (2002).
Improved local search for CP toolkits.
\emph{Annals of Operations Research}, 115, 31--50.

\bibitem[Thornton(2005)]{ThorntonJ2005}
Thornton, J. (2005).
Clause weighting local search for SAT.
\emph{Journal of Automated Reasoning}, 35, 97--142.

\bibitem[Umetani \& Yagiura(2007)]{UmetaniS2007}
Umetani, S., \& Yagiura, M. (2007).
Relaxation heuristics for the set covering problem.
\emph{Journal of the Operations Research Society of Japan}, 50, 350--375.

\bibitem[Umetani et~al.(2013)]{UmetaniS2013}
Umetani, S., Arakawa, M., \& Yagiura, M. (2013).
A heuristic algorithm for the set multicover problem with generalized upper bound constraints.
\emph{Proceedings of Learning and Intelligent Optimization Conference (LION)}, 75--80.

\bibitem[Umetani(2015)]{UmetaniS2015}
Umetani, S. (2015).
Exploiting variable associations to configure efficient local search in large-scale set partitioning problems.
\emph{Proceedings of AAAI Conference on Artificial Intelligence (AAAI)}, 1226--1232.

\bibitem[Van Hentenryck \& Michel(2005)]{VanHentenryckP2005}
Van Hentenryck, P., \& Michel, L. (2005).
\emph{Constraint-Based Local Search}, Cambridge: The MIT Press.

\bibitem[Voudouris et~al.(2001)]{VoudourisC2001}
Voudouris, C., Dorne, R., Lesaint, D., \& Liret, A. (2001).
iOpt: A software toolkit for heuristic search methods.
\emph{Proceedings of Principles and Practice of Constraint Programming (CP)}, 716--729.

\bibitem[Wedelin(1995)]{WedelinD1995}
Wedelin, D. (1995).
An algorithm for large-scale 0-1 integer programming with application to airline crew scheduling.
\emph{Annals of Operations Research}, 57, 283--301.

\bibitem[Yagiura et~al.(1999)]{YagiuraM1999}
Yagiura, M. \& Ibaraki, T. (1999).
Analysis on the 2 and 3-flip neighborhoods for the MAX SAT.
\emph{Journal of Combinatorial Optimization}, 3, 95--114.

\bibitem[Yagiura et~al.(2006)]{YagiuraM2006}
Yagiura, M., Kishida, M., \& Ibaraki, T. (2006).
A 3-flip neighborhood local search for the set covering problem.
\emph{European Journal of Operational Research}, 172, 472--499.
\end{thebibliography}
\end{document}